%% file: conv_20_onecolumn.tex
\documentclass[journal, onecolumn, draftcls]{IEEEtran}
\ifCLASSINFOpdf
\else
\fi
\usepackage{microtype}
\usepackage{graphicx}
\usepackage{subfigure}
\usepackage{booktabs,pifont} 

\usepackage{empheq}

\usepackage{hyperref}

\usepackage{pgfplots}
\pgfplotsset{compat=1.12}
\usetikzlibrary{shapes,arrows}
\usetikzlibrary{positioning}
\usepackage{amsmath,bm,times}
\usepackage{caption}
\usetikzlibrary{calc}
\usepackage{dsfont}
\usepackage{cuted}

\usepackage{mathtools}

\usepackage{amsthm}
\usepackage{comment}
\usepackage{enumitem}
\usepackage{arydshln}
\usepackage{multirow}
\usepackage{calc}
\usepackage{blkarray}
\usepackage{url}
\theoremstyle{definition}
\newtheorem{theorem}{Theorem}

\newtheorem{lemma}{Lemma}

\newtheorem{corollary}{Corollary}
\newtheorem{example}{Example}
\newtheorem{remark}{Remark}
\newtheorem{definition}{Definition}
\usepackage{graphicx}
\usepackage{dblfloatfix}
\usetikzlibrary{decorations.pathreplacing}
\usepackage{blindtext, graphicx, amsfonts,
	amssymb,multirow,epstopdf}
\def\BibTeX{{\rm B\kern-.05em{\sc i\kern-.025em b}\kern-.08em
    T\kern-.1667em\lower.7ex\hbox{E}\kern-.125emX}}

\makeatletter
\renewcommand*\env@matrix[1][*\c@MaxMatrixCols c]{%
  \hskip -\arraycolsep
  \let\@ifnextchar\new@ifnextchar
  \array{#1}}
\makeatother

\newcommand{\threevdots}{%
  \vbox{\baselineskip1ex\lineskiplimit0pt%
  \hbox{.}\hbox{.}\hbox{.}}}

\newcommand\bovermat[2]{%
    \makebox[0pt][l]{$\smash{\overbrace{\phantom{%
                    \begin{matrix}#2\end{matrix}}}^{\text{#1}}}$}#2}

\newcommand{\calT}{\mathcal{T}}

\newcommand{\calI}{\mathcal{I}}

\newcommand{\bfU}{\mathbf{U}}
\newcommand{\bfC}{\mathbf{C}}

\newcommand{\bfy}{\mathbf{y}}

\newcommand{\bfZ}{\mathbf{Z}}
\newcommand{\bfD}{\mathbf{D}}
\newcommand{\bfE}{\mathbf{E}}

\newcommand{\bfX}{\mathbf{X}}
\newcommand{\bfY}{\mathbf{Y}}

\newcommand{\bfA}{\mathbf{A}}

\newcommand{\bfG}{\mathbf{G}}

\newcommand{\bfu}{\mathbf{u}}

\newcommand{\bfx}{\mathbf{x}}

\newcommand{\bfc}{\mathbf{c}}
\newcommand{\bfB}{\mathbf{B}}
\newcommand{\bfI}{\mathbf{I}}
\newcommand{\bfz}{\mathbf{z}}
\newcommand{\bfM}{\mathbf{M}}
\newcommand{\bfR}{\mathbf{R}}

\newcommand{\bfL}{\mathbf{L}}

\begin{document}
%
\title{Efficient and Robust Distributed Matrix Computations via Convolutional Coding}

\definecolor{mygr}{rgb}{0.6,0.4,0.0}
\definecolor{my1color}{rgb}{0.6,0.4,0.0}
\definecolor{mycolor1}{rgb}{0.00000,0.44700,0.74100}%
\definecolor{mycolor2}{rgb}{0.85000,0.32500,0.09800}%
\tikzset{
block/.style    = {draw, thick, rectangle, minimum height = 2em, minimum width = 2em},
sum/.style      = {draw, circle, node distance = 1cm},
input/.style    = {coordinate},
output/.style   = {coordinate},
}


\author{\IEEEauthorblockN{Anindya Bijoy Das, Aditya Ramamoorthy and
Namrata Vaswani} \\
\IEEEauthorblockA{Department of Electrical and Computer Engineering,\\
Iowa State University, Ames, IA 50011 USA.\\
$\{$abd149,adityar,namrata$\}$@iastate.edu
}
\thanks{
This work was supported in part by the National Science
Foundation (NSF) under Grant CCF-1718470 and Grant CCF-1910840.}}

%



\IEEEtitleabstractindextext{%
\begin{abstract}
Distributed matrix computations -- matrix-matrix or matrix-vector multiplications -- are well-recognized to suffer from the problem of stragglers (slow or failed worker nodes). Much of prior work in this area is (i) either sub-optimal in terms of its straggler resilience, or (ii) suffers from numerical problems, i.e., there is a blow-up of round-off errors in the decoded result owing to the high condition numbers of the corresponding decoding matrices. Our work presents convolutional coding approach to this problem that removes these limitations. It is optimal in terms of its straggler resilience, and has excellent numerical robustness as long as the workers' storage capacity is slightly higher than the fundamental lower bound. Moreover, it can be decoded using a fast peeling decoder that only involves add/subtract operations. Our second approach has marginally higher decoding complexity than the first one, but allows us to operate arbitrarily close to the lower bound. Its numerical robustness can be theoretically quantified by deriving a computable upper bound on the worst case condition number over all possible decoding matrices by drawing connections with the properties of large Toeplitz matrices.  All above claims are backed up by extensive experiments done on the AWS cloud platform.
\end{abstract}

\begin{IEEEkeywords}
Distributed computing, Straggler, Convolutional coding, Toeplitz matrix, Vandermonde matrix.
\end{IEEEkeywords}}

\maketitle

\IEEEdisplaynontitleabstractindextext

%
\IEEEpeerreviewmaketitle

\section{Introduction}
%
%
%
%

\label{sec:intro}
Distributed computing clusters are heavily used in domains such as machine learning where datasets are often so large that they cannot be stored in a single computer. The widespread usage of such clusters presents several opportunities and advantages over traditional computing paradigms. However, they also present newer challenges. Large scale clusters which can be heterogeneous in nature suffer from the issue of stragglers (slow or failed workers in the system). Fig. \ref{strtime} shows the variation of speed of different {\tt t2.micro} machines in AWS (Amazon Web Services) cluster, and it can be seen that for a particular job, a slow worker node may require around $40\% -50\%$ more time than the average.


\input{straggler}

The conventional approach \cite{zaharia2008improving} to tackle stragglers has been to run multiple copies of tasks on various machines, with the hope that at least one copy finishes on time. For instance, consider matrix-vector multiplication with a matrix $\bfA$ and vector $\bfx$, where our goal is to obtain the product $\bfA^T \bfx$ in a distributed fashion. Fig. \ref{matvecfig} shows an example where we partition $\bfA$ into four block-columns and we assign two block-columns to each of the four worker nodes. Thus each block column has been assigned twice over all four workers and we can verify that we recover the final result if {\it any} three workers finish their respective jobs. In other words, we can say that this scheme is resilient to one straggler.

However, this toy example can be made even more efficient in terms of resource utilization by dividing $\bfA$ into two block-columns $\bfA_0$ and $\bfA_1$ and assigning the worker nodes appropriate linear combinations of $\bfA_0$ and $\bfA_1$ so that the required result can be decoded from any {\it two} workers. This is the basic idea underlying ``coded computation" (introduced in the work of Lee et al. \cite{lee2017high}).
It leverages ideas from erasure coding to introduce redundancy in the computation performed by the worker nodes. Roughly speaking, as long as enough worker nodes complete their tasks, the master node can decode the intended result by appropriate post-processing.

The central problem within coded distributed matrix computation can be explained as follows. Suppose that we have large matrices $\bfA \in \mathbb{R}^{t\times r}, \bfB \in \mathbb{R}^{t \times w}$ and a vector $\bfx \in \mathbb{R}^{t}$. The goal is to either compute $\bfA^T \bfx$ (matrix-vector multiplication) or $\bfA^T \bfB$ (matrix-matrix multiplication) in a distributed fashion using $n$ worker nodes while being resistant to any $s$ stragglers. Redundancy is introduced in the computation by coding across appropriately chosen submatrices of $\bfA$ and $\bfB$ and assigning the worker nodes appropriate computation responsibilities.

The main finding of several recent works in this area is that it is possible to embed distributed matrix computations into the structure of an equivalent erasure code, where the failed nodes play the role of erasures \cite{lee2018speeding,yu2020straggler,yu2017polynomial,dutta2016short,dutta2019optimal,mallick2018rateless, wang2018coded} (we discuss related work in detail shortly). A given coded computation scheme is said to have threshold $\tau$ if the desired result can be decoded as long as any $\tau$ worker nodes return their results to the master node. This has been the focus of many works in the literature.

In this work, we consider the important issue of numerical stability within coded computation (in addition to threshold). We point out that several of the existing schemes in the literature suffer from significant numerical issues in the decoding process. In particular, the system of equations that is solved by the master node in the decoding step can have a very high condition number which in turn results in a large error in the decoded result. We present a novel scheme based on convolutional codes (operating over the reals) that simultaneously addresses numerical stability, the threshold, and possesses easy encoding/decoding. An overview of the properties of most of the known schemes in the literature is presented in Table \ref{comparison}.

\input{matvecfig}

\newcommand{\cmark}{\ding{51}}%
\newcommand{\xmark}{\ding{55}}

\begin{table*}[t]
\caption{{\small Comparison with existing works \cite{mallick2018rateless,lee2017high,yu2017polynomial,8849468} and parallel works \cite{8919859,ramamoorthy2019numerically} in terms of different properties of the algorithms. Decoding complexity is mentioned for $s$ stragglers with recovery threshold $k$ where $\bfA \in \mathbb{R}^{t \times r}$ and $\bfB \in \mathbb{R}^{t \times w}$. $T$ and $q$ are decoding algorithm parameters for the random conv. code, discussed in Section \ref{sec:cond_no_discussion}, where $T, q \ll r , w$.}} 
\label{comparison}

\begin{center}
\begin{small}
\begin{sc}
\begin{tabular}{c c c c c}
\hline
\toprule
\multirow{2}{1 cm}{Codes} & Mat-Mat & Optimal & Numerical & Decoding Complexity\\
 & Mult? & Threshold? & Stability? & for Mat-Mat Mult\\
 \midrule
Repetition Codes & \cmark & \xmark & \cmark &  Zero  \\ \hline
Rateless Codes \cite{mallick2018rateless}  & \xmark & \xmark & \cmark &  \xmark \\ \hline
Product Codes \cite{lee2017high} & \cmark &\xmark & \xmark  & $O(r^3)$, assuming $r = w$\\ \hline

Polynomial Codes \cite{yu2017polynomial} & \cmark  & \cmark & \xmark & $O(rwk)$  \\ \hline
Ortho-Poly Codes \cite{8849468} & \cmark & \cmark & \cmark & $O(rwk)$\\ \hline
Circulant and Rotation Matrix \cite{ramamoorthy2019numerically} & \cmark  & \cmark & \cmark & $O(rwk)$  \\ \hline
Random Khatri-Rao Codes \cite{8919859} & \cmark & \cmark & \cmark & $O\left( \frac{rw}{k} s^2 \right)$\\ \hline
\textbf{All-Ones-Conv Code (Proposed)} & \cmark & \cmark & \cmark & $O(rws)$ {\footnotesize (add/subtract ops)} \\ \hline
\textbf{Random-Cov Code (Proposed)} & \cmark & \cmark & \cmark & $ \min(T,q) \times O\left( \frac{rw}{k} s^2 \right)$\\
%
\bottomrule
\end{tabular}
\end{sc}
\end{small}
\end{center}
\end{table*}%

This paper is organized as follows. Section \ref{sec:prob_form} explains the problem formulation and Section \ref{sec:lit} describes the background and related work and summarizes of the contributions of our work. Section \ref{sec:conv_code_intro} discusses our main ideas on how convolutional codes can be used to address distributed matrix computations, Section \ref{sec:cond_no_discussion} overviews the analysis of numerical stability for our codes and Section \ref{sec:numerical_exp} discusses the experimental performance of our proposed methods and shows the comparison with other available approaches. We conclude the paper with a discussion about future work in Section \ref{sec:conclusion}. For the sake of readability several of the proofs appear in the Appendix.

\section{Problem Formulation}
\label{sec:prob_form}
In the matrix-vector case we partition $\bfA$ into submatrices of equal size and $\bfx$ into subvectors and distribute a certain number of ``coded'' versions of these submatrices to the $n$ workers (subject to a storage constraint). Every worker computes the product of its assigned submatrices and subvectors  and sends the computed result back to the master node. The master then ``decodes'' to recover $\bfA^T \bfx$.

In the matrix-matrix multiplication scenario, each worker node receives coded versions of submatrices of $\bfA$ and  coded versions of the submatrices of $\bfB$ \footnote{A general formulation need not restrict the assignment to coded submatrices of $\bfA$ and $\bfB$. Nevertheless, all known schemes thus far and our proposed schemes work with equal-sized submatrices, so we present the formulation in this way.}. It computes pairwise products (either all or some subset thereof) of these and sends them to the master node which needs to decode to recover $\bfA^T \bfB$.

In the discussion below we discuss the matrix-matrix scenario; it applies in a natural way to the matrix-vector case as well. We consider a $p \times u$ and $p \times v$ block decomposition of $\bfA$ and $\bfB$ respectively as shown below.
\begin{align*}
\bfA = \begin{bmatrix}
\bfA_{0,0} &\dots& \bfA_{0,u-1}\\
\vdots & \ddots & \vdots \\
\bfA_{p-1,0} & \dots & \bfA_{p-1,u-1}
\end{bmatrix}  ;  \;\;\; \textrm{and} \; \; \; \bfB = \begin{bmatrix}
\bfB_{0,0} &\dots& \bfB_{0,v-1}\\
\vdots & \ddots & \vdots \\
\bfB_{p-1,0} & \dots & \bfB_{p-1,v-1}
\end{bmatrix}. 
\end{align*} The master node encodes by computing appropriate scalar linear combinations of the $\bfA_{i,j}$ matrices and respectively the $\bfB_{i,j}$ submatrices. This implies that the master node only performs scalar multiplications and additions. It is not responsible for any of the computationally intensive matrix operations.  Following this, it sends the corresponding coded submatrices to each of the workers.


We assume that a worker node cannot store the whole matrix $\bfA$ or $\bfB$. Each worker can store the equivalent of $\gamma_A$ fraction of matrix $\bfA$ and $\gamma_B$ fraction of matrix $\bfB$; this is referred to as the storage fraction.

The assumption is that some nodes will fail or will be too slow, the maximum number of such nodes is assumed to be $s$ or less. The goal is to design the coding scheme so that (i) the decoding is possible using the output of any $k = (n-s)$ workers ($k$ is often called the recovery threshold of the scheme),  (ii) it is robust to noise (both numerical precision errors and other sources of noise); and (iii) it is efficiently decodable. We say that the threshold of a scheme is {\it optimal} if it is the lowest possible given the storage constraints.


\section{Background, Related Work and Summary of Contributions}
\label{sec:lit}

In recent years, several coded computation schemes have been proposed for matrix multiplication \cite{lee2018speeding, yu2020straggler, yu2017polynomial, dutta2016short, dutta2019optimal, mallick2018rateless, wang2018coded, c3les, TangKR19,  kiani2018exploitation}. We illustrate the basic idea below using the polynomial code approach of \cite{yu2017polynomial}. These ideas are presented in a tutorial fashion in \cite{ramamoorthyDTMag20}.

Consider a scenario with $n = 5$ workers where each of these worker nodes can store $\gamma_A = \frac{1}{2}$ fraction of matrix $\bfA$ and $\gamma_B = \frac{1}{2}$ fraction of matrix $\bfB$. Consider $u = v = 2$ and $p = 1$, thus we partition both $\bfA$ and $\bfB$ into two block-columns $\bfA_0, \bfA_1$ and $\bfB_0, \bfB_1$ respectively. Next, we define two matrix polynomials as
\begin{align*}
\bfA(z) &= \bfA_0 + \bfA_1 z  \; \; \; \textrm{and} \; \; \; \bfB(z) = \bfB_0 + \bfB_1 z^2  ;\\
\textrm{so} \; \bfA^T(z) \bfB(z) & = \bfA^T_0 \bfB_0 + \bfA_1^T \bfB_0 z + \bfA_0^T \bfB_1 z^2 + \bfA_1^T \bfB_1 z^3 .
\end{align*} The master node evaluates these polynomial $\bfA(z)$ and  $\bfB(z)$ at distinct real values $z_0, z_1, \dots, z_{n-1}$, and sends the corresponding matrices to worker node $W_i$ (see Fig. \ref{matmatfig} where $z_i = i+1$). Each worker node computes the product of its assigned submatrices. 
It follows that decoding at the master node is equivalent to decoding a degree-3 real-valued polynomial. Thus, the master node can recover $\bfA^T \bfB$ as soon as it receives the results from {\textit any} four workers. Thus, in this example, the recovery threshold is, $k = 4$ and the system is resilient to $s = 1$ straggler.

\input{matmatfig}

A different solution can be obtained using the approach in \cite{dutta2019optimal} for the same example. Let $u = v = 1$ and $p = 2$, so we can write $\bfA^T \bfB = \bfA_0^T \bfB_0 + \bfA_1^T \bfB_1$. Now we define two matrix polynomials as
\begin{align*}
  \bfA(z) &= \bfA_0 z + \bfA_1  \; \; \; \textrm{and} \; \; \; \bfB(z) = \bfB_0 + \bfB_1 z ; \\
\textrm{so} \; \bfA^T(z) \bfB(z) & = \bfA^T_1 \bfB_0 + \left( \bfA_0^T \bfB_0 + \bfA_1^T \bfB_1 \right) z + \bfA_0^T \bfB_1 z^2  .
\end{align*} As before, the master node will evaluate the polynomial $\bfA(z)$ and  $\bfB(z)$ at $z_0, z_i, \dots, z_{n-1}$, and send the corresponding matrices to worker node $W_i$. It follows that the master can recover all the unknowns $\left(\textrm{including} (\bfA_0^T \bfB_0 + \bfA_1^T \bfB_1 ) \right)$ as soon as it receives the results from {\it any} three workers. Thus, in this example, the recovery threshold is, $k = 3$ and the system is resilient to $s = 2$ stragglers.

It should be noted that the latter approach can lead to more straggler resilience, but the computational load per worker has doubled compared to the first approach. Moreover the communication load from the worker nodes to the master node is also higher by a factor of $4$ compared to the first approach.

For both schemes above, it can be shown that worker node computation time depends on $t$, whereas the decoding complexity is independent of it (see for instance \cite{ramamoorthyDTMag20}). Thus, for scenarios where $t$ is very large, the decoding time can be neglected. Nevertheless, a low decoding complexity is desirable from a practical standpoint.


\subsection{Related Work}
\label{relwork}

As discussed above, \cite{yu2020straggler,yu2017polynomial, dutta2019optimal} convert distributed matrix computation into polynomial evaluation/interpolation, i.e., the coded submatrices correspond to polynomial evaluation maps. We remark here that as far as we are aware, the idea of embedding matrix multiplication using polynomial maps goes back even further to Yagle \cite{365287} (the motivation there was fast matrix multiplication).


For fixed storage constraints $\gamma_A = \frac{1}{u}$ and $\gamma_B = \frac{1}{v}$ and for fixed computation overhead per worker with $p=1$ and arbitrary $u$ and $v$, the optimal threshold $\tau$ is shown to be $uv$ \cite{yu2017polynomial} using the polynomial approach. When $p \geq 2$, the work of \cite{yu2020straggler} demonstrates a threshold of $puv + p-1$. They also present a converse argument which demonstrates that this is within a factor of two of the optimal threshold. 

While the computation threshold is somewhat well understood at this point, the issue of numerical stability has received much less attention. When operating over finite fields, proving the invertibility of an appropriate submatrix of the coding matrix suffices to guarantee correct decoding. However, in decoding a real system of equations, errors in the input can get amplified by the condition number (ratio of maximum and minimum singular values) of the associated matrix; hence, a low condition number is critical.
For instance, in solving a square system of equations $\bfy = \bfM \bfx$, suppose that $\bfy$ is perturbed to $\tilde{\bf{y}}$ (owing to round-off errors) and that the estimate of $\bf{x}$ is $\hat{\bf{x}}:= \bf{M}^{-1} \tilde{\bf{y}}$. Then, the normalized error in $\hat{\bf{x}}$ is given by
\begin{align*}
\frac{\|\hat{\bf{x}} - \bf{x}\|}{\|\bf{x}\|} =  \frac{ \|\bf{M}^{-1} (\tilde{\bf{y}} - \bf{y})\| }{ \|\bf{M}^{-1} \bf{y}\| } \le
 \frac{\sigma_{\max}(\bf{M}^{-1} )}{\sigma_{\min}(\bf{M}^{-1})} \frac{\|\tilde{\bf{y}} - \bf{y}\|}{\|\bf{y}\|} =  \frac{\sigma_{\max}(\bf{M})}{\sigma_{\min}(\bf{M})} \frac{\|\tilde{\bf{y}} - \bf{y}\|}{\|\bf{y}\|} =  \kappa(\bf{M}) \frac{\|\tilde{\bf{y}} - \bf{y}\|}{\|\bf{y}\|} ,
\end{align*}
where $ \sigma_{\max}(\bfM)$ and $\sigma_{\min}(\bfM)$ denote the maximum and minimum singular values of $\bfM$ and their ratio $\kappa(\bf{M})$ is the condition number of the decoding matrix $\bfM$. Thus, it is clear that a small condition number of the decoding matrix leads to less amplification of the round-off error in $\hat{\bf{x}}$.

This issue is especially relevant since it is well recognized that polynomial interpolation over the reals suffers from significant numerical issues since the corresponding Vandermonde matrices have very high condition numbers (that are exponential in their size \cite{Pan16}).  In fact, even for clusters with around $n=30$ nodes, the condition number of the polynomial approach \cite{yu2017polynomial} is so large that the decoded result is essentially useless (see Section \ref{sec:numerical_exp}). We note here that Section VII of \cite{yu2020straggler} remarks that the numerical issues can be handled by embedding all operations within a finite field. In Section \ref{sec:numerical_exp}, we demonstrate that the performance of this method is strongly dependent on the entries of matrices $\bfA$ and $\bfB$ and the resultant normalized MSE can be quite bad \cite{Tang2020Auth}.

Some recent works have highlighted and considered the issue of numerical stability in this context. The work of \cite{8849395, 8849451} presented strategies for distributed matrix-vector multiplication and demonstrated some schemes that empirically have better numerical performance than polynomial based schemes for some values of $n$ and $s$. The work in \cite{8849395}  considers a convolutional coding approach, but from a parity check matrix perspective and the work in \cite{8849451} uses universally decodable matrices which further allows to utilize the partial computations of the stragglers. However, both these approaches work only for the matrix-vector problem and do not provide a computable bound on the condition number of the decoding submatrices.


The work of \cite{8849468} presents an alternate approach that works within the basis of orthogonal polynomials. They demonstrate that the worst case condition number of their schemes is at most $O(n^{2s})$ and their numerical experiments demonstrate improvements with respect to \cite{yu2017polynomial}. Our experimental evaluation in Section \ref{sec:numerical_exp} clearly demonstrates that our proposed schemes have condition numbers that are orders of magnitude lower than \cite{8849468}.  \cite{8919859} present an approach where the encoded matrices are generated by taking random linear combinations of the block-columns of the respective matrices (this was also suggested in Remark $8$ of \cite{yu2017polynomial}). We note here that their approach can be considered as a subclass of our methods, as discussed in Section \ref{sec:numerical_exp}. Table \ref{comparison} shows a comparison of the features of several well-known approaches for distributed matrix computations. Our results in Section \ref{sec:numerical_exp} show that the underlying structure of our codes consistently results in lower worst case condition numbers than \cite{8919859}. Finally, the parallel work of \cite{ramamoorthy2019numerically} presents an approach that leverages the properties of rotation matrices and circulant permutation matrices. They demonstrate that the worst case condition number of their recovery matrices grow at most as $O(n^{s+6})$. While their numerical results are better than ours, our work has the advantage of easy encoding and decoding and explores a convolutional approach to this problem which has not been considered before.

\subsection{Summary of Contributions}
In this paper we present an efficient and robust scheme for coded matrix computations that is inspired by convolutional codes. Our codes operate over the reals, unlike the majority of convolutional codes that are considered over finite fields \cite{lincostello}. Crucially, they exploit the Vandermonde property of the recovery matrices, where the matrices are defined over a different field (formal Laurent series over $\mathbb{R}$) than the real numbers. This naturally allows for simple encoding and decoding in addition to ensuring the threshold properties.


\begin{itemize}[wide, labelwidth=!, labelindent=0pt]
    \item  Our work is among the first to provide an efficient coded computation approach for {\em both} matrix-vector and matrix-matrix multiplications that provably works in the (i) essentially noise-free regime where numerical precision issues dominate, and (ii) the noisy regime where noise is significant.

    \item We present two classes of codes in this work. Our first approach can be decoded using a peeling decoder using only add/subtract operations and has excellent numerical performance when the storage capacity of the nodes is slightly higher than the fundamental lower bound.

    When operating very close to the storage capacity lower bound, we propose an alternative random convolutional coding strategy for which we can provide a ``computable'' upper bound ({\it cf.} Theorem \ref{theorem:regalia_stuff} in Section \ref{sec:upper_bd_kappa}) on the worst case condition number of the recovery matrices. This naturally leads to a random sampling algorithm to pick a coding matrix with good performance. Our work draws novel connections with this problem and the asymptotic analysis of large Toeplitz matrices \cite{gray2006toeplitz}.

    \item An exhaustive comparison of our work with other approaches in the literature shows that the condition numbers of our work are orders of magnitude below all the comparable approaches (except \cite{ramamoorthy2019numerically}) and have fast decoding times. Fig. \ref{error_18s3_intro} depicts a comparison of the performance of the different schemes considered in our work. 

    \item As far as we are aware, most previous work has approached coded computation by exploiting its link with block codes under erasures. Our work is the first to investigate a convolutional coding approach to this problem. This in turn opens up newer problems for investigation in this area.
\end{itemize}

\input{err_matmat}

\section{Convolutional Coding for Distributed Matrix Computation}
\label{sec:conv_code_intro}

\subsection{Simple Illustrative Example}
\label{sec:simple_illus}
We explain our key idea by means of the following example.
Consider two row vectors in $\mathbb{R}^q$, $\mathbf{u}_0 = [u_{00}~ u_{01}~ \dots~ u_{0(q-1)}]$ and $\mathbf{u}_1 = [u_{10}~ u_{11}~ \dots~ u_{1(q-1)}]$. These vectors can also be represented as polynomials in the {\it indeterminate} $D$, $\mathbf{u}_i(D) = \sum\limits_{j=0}^{q-1}u_{ij} D^j$ for $i = 0,1$. As explained in Appendix \ref{sec:invertibility_mat}, these polynomials can be treated as elements in the ring of formal Laurent series in $D$ \cite{niven1969formal}. Moreover, it can be shown that this ring is in fact a field, i.e., each element has a corresponding inverse.
Consider the following encoding of $[\mathbf{u}_0(D) \; \; \mathbf{u}_1(D)]$.
\begin{align*}
[\bfc_0(D) \;\; \bfc_1(D) \;\; \bfc_2(D) \;\; \bfc_3(D)]
 = \;  \left[\mathbf{u}_0(D) \; \mathbf{u}_1(D) \right] \;
\underbrace{\begin{bmatrix}
1 & 0 & 1 & 1 \\
0 & 1 & 1 & D \\
\end{bmatrix}.
}_{\bfG(D)}
\end{align*}

It is not too hard to see that the polynomials $\mathbf{u}_0(D)$ and $\mathbf{u}_1(D)$ (equivalently the vectors $\mathbf{u}_0, \mathbf{u}_1$) can be recovered (or ``decoded'') from any two entries of the vector $[\bfc_0(D) ~\bfc_1(D)~\bfc_2(D)~\bfc_3(D)]$. For instance, suppose that we only receive $\bfc_2(D)$ and $\bfc_3(D)$. Notice that
\vspace{-0.1 in}
\begin{align*}
 \bfc_2(D) \;  & = \;  \sum_{j=0}^{q-1} (u_{0j} + u_{1j})D^j \; \;\; \textrm{and} \\
 \bfc_3(D) \;  & = \;  u_{00} + \sum_{j=0}^{q-2} (u_{0(j+1)} + u_{1j})D^j + u_{1(q-1)}D^q.
\end{align*}

Starting with $u_{00}$ from the constant term of $\bfc_3(D)$, one can iteratively recover each of the coefficients of $\mathbf{u}_0(D)$ and $\mathbf{u}_1(D)$, with only one new variable to recover in each iteration. A similar argument applies if we consider a different set of two entries from $[\bfc_0(D) ~\bfc_1(D)~\bfc_2(D)~\bfc_3(D)]$. We refer to such a decoding scheme as a ``peeling decoder''.

Observe that the encoded polynomial $\bfc_3(D)$ has degree $q$, while the others have degree $q-1$. Thus, if the coefficients of the polynomials $\bfc_i$ correspond to encoded data that were sent to node $i$ for processing, then node 3 would need slightly higher storage/processing capacity than nodes 0, 1, 2. Secondly, observe that the above idea can also be equivalently understood by replacing the $2 \times 4$ matrix of polynomials $\bfG(D)$ by a larger matrix of size $2q \times (4q+1)$ and rewriting all the scalar polynomials as row vectors. Let $\bfc_0, \bfc_1, \bfc_2$ be row vectors of length $q$ and $\bfc_3$ be a row vector of length $q+1$. Then,

\begin{align*}
& \begin{bmatrix}
\bfc_0 \;\; \bfc_1 \;\; \bfc_2 \;\; \bfc_3
\end{bmatrix} =  \begin{bmatrix}
\bfu_0 \;\; \bfu_1
\end{bmatrix} \; \begin{bmatrix}
\bfI_q & \mathbf{0}_{q\times q} & \bfI_q \;\;\; [\bfI_{q} \;\;\mathbf{0}] \\
\mathbf{0}_{q\times q} & \bfI_q & \bfI_q \;\;\;  [\mathbf{0}\;\; \bfI_q]
\end{bmatrix}
\end{align*}where $\mathbf{0}_{q \times q}$ is  a $q \times q$ matrix of zeroes, $\bfI_q$ is a $q \times q$ identity matrix, and $\mathbf{0}$ is a column of zeroes. In what follows, we consider generalizations of this basic example where the $\bfu_i$'s will correspond to block-columns of $\bfA$ and $\bfB$.

\subsection{Proposed matrix-vector multiplication scheme}
\label{sec:matvec_section}
The above idea can naturally be adapted to the distributed matrix-vector multiplication setting. We show an example in Fig. \ref{matvec} with $n = 4$ workers and $s = 2$ stragglers, so $k = n - s = 2.$. Suppose that matrix $\bfA$ is partitioned into $k q$ block-columns (the choice of $q$ will be discussed shortly). In our work, the presentation follows more naturally if we index the block-columns of $\bfA$ using two indices instead of one. In particular, they are indexed as $\bfA_{\langle i,j \rangle}, i \in [k], j \in [q]$ (where $[\ell]$ denotes the set $\{0, \dots, \ell -1\}$) and each worker node stores at most $\gamma r$ columns of length-$t$ ($\gamma$ is called the storage fraction).

Let $\bfU_i(D) = \sum_{j=0}^{q-1} \bfA^T_{\langle i,j \rangle} D^{j}$ for $0 \leq i \leq k-1$. Furthermore, let $\bfY_{k,s}$ denote a $k \times s$ matrix whose $(i,j)$-th submatrix is  $(\bfY_{k,s})_{i,j} = (D^j)^{i}$, for $i \in [k], j \in [s]$, i.e., $\bfY_{k,s}$ has the Vandermonde structure. We define
\begin{align}
\label{GmvD}
\bfG_{mv}(D) \; = \; \begin{bmatrix}
     \underbrace{\bfI_k}_{\textrm{message part}} \; \;\bigg{|} \; \; \underbrace{\bfY_{k,s}(D)}_{\textrm{parity part}}
\end{bmatrix}.
\vspace{-0.05in}
\end{align}
Consider the encoding
\begin{align*}
 [\bfC_0(D) \;\;\; ~\bfC_1(D)~ \; \dots ~ \; \bfC_{n-1}(D)] \nonumber = \;  [\bfU_0(D) \;\;\; ~ \bfU_1(D)~ \; \dots~ \; \bfU_{k-1}(D)]~ \bfG_{mv}(D).
\end{align*}

To arrive at the distributed matrix-vector multiplication scheme, we simply interpret the coefficients of the powers of $D$ in $\bfC_i(D)$ as the encoded submatrices assigned to worker $i$ (see Fig. \ref{matvec} for an example). With this assignment, worker $i$ computes the inner product of its assigned matrices and $\bfx$. We say that a $k \times n$ matrix is maximum-distance-separable (MDS) if any of its $k \times k$ submatrices is nonsingular. This property further implies that $\bfA^T \bfx$ can be recovered as long as any $k$ workers complete their tasks. The following result shows that $\bfG_{mv}(D)$ is MDS; the proof appears in the Appendix.

\begin{corollary}[Corollary of upcoming Theorem \ref{theorem:nonsingular_G} given in Section \ref{sec:matmat_section}]
\label{thm1_cor}
Any $k \times k$ submatrix of $\bfG_{mv}(D)$ has a determinant which is a non-zero polynomial in $D$, i.e., it is non-singular.
\end{corollary}

Analogous to convolutional coding, we call the first $k$ workers the message workers and the last $s$ workers the parity workers.
Each of the first $k$ message workers receives $q$ submatrices $\bfA_{\langle i,j \rangle}, j=0,1,\dots,q-1$, each of which is a matrix of size  $t \times r/(kq)$. The rest of the $s$ parity workers will receive $ \geq q$ such submatrices.
The highest exponent of $D$ in the generator matrix $\bfG_{mv}(D)$ is $(s-1)(k-1)$. Thus, the maximum storage needed by a worker is $q+(s-1)(k-1)$ submatrices. When $q$ is large enough, this imbalance is not significant. If we assume a bound of $\gamma$ on the storage capacity fraction of any worker, we need

\begin{align}
\bigg( q + (s-1)(k-1) \bigg) \frac{r}{kq} \; &\leq \; \gamma r, \nonumber\\
\implies q \geq \frac{(s-1)(k-1)}{k(\gamma - \frac{1}{k})}. \label{eq:lower_bd_Delta}
\end{align}

\label{eg:comps}

For example, in Fig. \ref{matvec}, $\gamma$ is set to $\frac{5}{8}$ which leads to $q = 4$.

\input{matvecex}

\subsection{Proposed matrix-matrix multiplication scheme}
\label{sec:matmat_section}
The matrix-matrix multiplication case requires the generalization of the above ideas.
Let $\bar{a} = [a_0~a_1~ \dots ~ a_{s-1}]$ and $\bar{b} = [b_0~b_1~\dots~b_{k-1}]$ be vectors of non-negative integers such that $0 \leq a_0 < a_1 < \dots < a_{s-1}$ and $0 \leq b_0 < b_1 < \dots < b_{k-1}$. Let $\bfY_{\bar{b}, \bar{a}}(D)$ denote a $k \times s$ matrix whose $(i,j)$-th entry is given by
\begin{align}
[\bfY_{\bar{b}, \bar{a}}(D)]_{i,j} = (D^{a_j})^{b_i}. \label{eq:gen_form_Y}
\end{align}
Using this matrix, define a generalization of $\bfG_{mv}(D)$ as follows
\vspace{-0.13in}
\begin{align}
\label{GnkD}
\bfG (D) = \begin{bmatrix}
    \; \bfI_k \; \;\; \big{|}  \;  \; \; \bfY_{\bar{b},\bar{a}}(D) \;
\end{bmatrix}.
\end{align}

Observe that we obtain $\bfG_{mv}(D)$ by setting  $a_j = j, 0 \leq j \leq s-1$ and $b_i = i, 0 \leq i \leq k-1$, which corresponds to $\bfY_{k,s}(D)$. We will design an encoding scheme for matrix-matrix multiplication whose equivalent generator matrix is of the form in \eqref{GnkD}. Before we explain the design, we show that this matrix also satisfies the MDS property (the proof appears in the Appendix).

\begin{theorem}
\label{theorem:nonsingular_G}
{\it Any} $k \times k$ submatrix of the generator matrix $\bfG(D)$ defined in \eqref{GnkD} is non-singular.
\end{theorem}
While non-singularity by itself does not reveal information about the corresponding condition numbers, Theorem \ref{theorem:nonsingular_G} provides a class of schemes with a specific structure that have excellent numerical stability (see Fig. \ref{error_18s3_intro} ``All Ones" curve) and can be modified and analyzed for condition number using the techniques discussed in Theorem \ref{theorem:regalia_stuff} within Section \ref{sec:cond_no_discussion}. The structure of $\bfG(D)$ in \eqref{GnkD} also allows for an efficient peeling decoder.

In the matrix-matrix case, we design generator matrices $\bfG_A(D)$ of size $k_A \times n$ and $\bfG_B(D)$ of size $k_B \times n$ such that $s = n - k_A k_B$. Each worker stores fractions $\gamma_A$ and $\gamma_B$ of matrices $\bfA$ and $\bfB$ respectively. Let $z$ be a large enough positive integer and let
\begin{align}
    \bfU^A_i(D) &= \sum_{j=0}^{q_A-1} \bfA^T_{\langle i,j \rangle} D^{zj}, i \in [k_A], \text{~and}\\
    \bfU^B_i(D) &= \sum_{j=0}^{q_B-1} \bfB_{\langle i,j \rangle} D^{j}, i \in [k_B].
\end{align}
Furthermore, we let
$\bfU^A(D) = [\bfU^A_0(D) ~\dots ~ \bfU^A_{k_A -1}(D)]$ and $\bfU^B(D) = [\bfU^B_0(D) ~\dots ~ \bfU^B_{k_B -1}(D)]$. The final goal of the master node is to recover all products of the form $\bfA^T_{\langle i_1,j_1 \rangle} \bfB_{\langle i_2,j_2 \rangle}$ for $i_1 \in [k_A], j_1 \in [q_A], i_2 \in [k_B], j_2 \in [q_B]$. Once again by forming
\begin{align*}
[\bfC^{A}_0(D) \;\; \bfC^{A}_1(D) ~\dots~ \bfC^{A}_{n-1}(D)] &= \bfU^A(D) \bfG_{A}(D), \text{~and}\\
[\bfC^{B}_0(D) \;\; \bfC^{B}_1(D) ~\dots~ \bfC^B_{n-1}(D)] &= \bfU^B(D) \bfG_{B}(D),
\end{align*} we can represent the assignment of coded submatrices of $\bfA$ and $\bfB$ to worker node $i$ by the coefficients of $\bfC^{A}_i(D)$ and $\bfC^{B}_i(D)$ respectively. Following this step, each worker node computes the pairwise product of each coded submatrix of $\bfA$ and coded submatrix of $\bfB$ assigned to it.

The matrices $\bfG_{A}(D)$ and $\bfG_{B}(D)$ will be picked in such a way so that the pairwise product of each coefficient of $\bfC^{A}_i(D)$ and each coefficient of $\bfC^{B}_i(D)$ appears in  $\bfC^{A}_i(D) \times \bfC^{B}_i(D)$, i.e., each worker node equivalently computes $\bfC^{A}_i(D) \times \bfC^{B}_i(D)$.
Using MATLAB notation and Kronecker product properties, for $i= 1, 2, \dots, n$, we have
\begin{align*}
\bfC^{A}_i(D) \times \bfC^{B}_i(D) = \; & \left[\bfU^A(D) \bfG_{A}(D)(:,i)\right] \times \left[\bfU^B(D) \bfG_{B}(D)(:,i)\right] \\
= \; & \left[\bfU^A(D) \otimes \bfU^B(D)\right] \times \left[\bfG_{A}(D)(:,i) \otimes \bfG_{B}(D)(:,i)\right],
\end{align*}
where $\otimes$ denotes the Kronecker product. Therefore, the computation peformed by the worker nodes can be compactly represented using the Khatri-Rao product \cite{zhang_matrix_anal} (denoted by $\odot$)\footnote{For two matrices with the same column dimension, the Khatri-Rao product corresponds to the matrix obtained by taking the Kronecker product of the corresponding columns.} Moreover, using the properties of the Khatri-Rao product, we have
\vspace{-0.05in}
\begin{align}
\label{uaub}
 \left[\bfU^A(D) \bfG_{A}(D) \right] \odot  \left[\bfU^B(D) \bfG_{B}(D)\right]
 = \;  \left[\bfU^A(D) \otimes \bfU^B(D) \right] \; \left[\bfG_{A}(D) \odot \bfG_{B}(D) \right] .
\end{align}

The key idea at this point is to ensure that $\bfG_{A}(D) \odot \bfG_{B}(D)$ has the structure of a matrix as in (\ref{GnkD}). Towards this end, we choose
\begin{align*}
\bfG_A(D) \; &= \; \begin{bmatrix}[cccc|cc]
\; \; \bovermat{$k_A$}{\mathbf{1}_{k_B}  & \mathbf{0} &\dots & \mathbf{0} \; & }  \\
\; \;\mathbf{0} & \mathbf{1}_{k_B}  &\dots & \mathbf{0} \;& \\
\; \;\mathbf{0} & \mathbf{0}  &\dots & \mathbf{0} \; & \; \bfY_{k_A,s}(D^z)\\
\; \;\threevdots & \threevdots & \dots & \threevdots \; &  \\
\; \;\mathbf{0} & \mathbf{0}  & \dots & \mathbf{1}_{k_B} \; &
\end{bmatrix}, \\\\
\bfG_B(D)\; &= \; \begin{bmatrix}[cccc|cc]
\; \; \bovermat{$k_A$}{\bfI_{k_B} & \bfI_{k_B} & \dots & \bfI_{k_B}} & \; \bfY_{k_B,s}(D) \;
\end{bmatrix},
\end{align*} where $\mathbf{1}_{k_B}$ is an all-ones row vector of length $k_B$, and the total number of rows in $\bfG_A(D)$ and $\bfG_B(D)$ are $k_A$ and $k_B$ respectively.
This implies that
\begin{equation}
\bfG_{A}(D) \odot \bfG_{B}(D) = [\bfI_{k} \; ~|~ \; \bfY_{k_A,s}(D^z) \odot  \bfY_{k_B,s}(D)]
\label{eq:khatri_eq}
\end{equation} where $k = k_A k_B$. The following lemma shows that the RHS of (\ref{eq:khatri_eq}) has the structure of the matrix in (\ref{GnkD}).

\begin{lemma}
\label{lemma:khatri}
The Khatri-Rao product $\bfY_{k_A,s}(D^z) \odot \bfY_{k_B,s}(D)$ is a matrix in the form of (\ref{eq:gen_form_Y}).
\vspace{-0.05 in}
\end{lemma}
\begin{proof}
Note that the Kronecker product of $\ell$-th column of $\bfY_{k_A,s}(D^z)$  and $\ell$-th column of $\bfY_{k_B,s}(D)$ can be expressed as
\begin{align}
\label{lem1eq}
\begin{bmatrix}
1\\
D^{zl}\\
D^{2zl}\\
\vdots\\
D^{(k_A -1)zl}
\end{bmatrix}
\otimes
\begin{bmatrix}
1\\
D^{l}\\
D^{2l}\\
\vdots\\
D^{(k_B -1)l}
\end{bmatrix}
= \begin{bmatrix}
1\\
\vdots\\
D^{(k_B -1)l}\\
D^{zl}\\
\vdots\\
D^{(k_B-1 +z)l}\\
\vdots\\
D^{(k_A-1)zl}\\
\vdots\\
D^{(k_B-1+ (k_A -1)z)l}\\
\end{bmatrix}
\end{align}
The vector on the RHS above consists of powers of $D^l$ and can be seen to be in the form of (\ref{eq:gen_form_Y}).
\end{proof}

Lemma \ref{lemma:khatri} explains why Theorem \ref{theorem:nonsingular_G} is applicable to the coding scheme used for matrix-matrix multiplication. Thus, this lemma, along with Theorem \ref{theorem:nonsingular_G} implies that the proposed convolutional code based matrix-matrix multiplication scheme is MDS.


Now, we need to choose such a value of $z$ which ensures that $\left[\bfU^A(D) \otimes \bfU^B(D)\right]$ in \eqref{uaub} contains all the distinct pairwise products that we are interested. We know that worker $i$ will be assigned the jobs according to the  column $i$ of the RHS in (\ref{eq:khatri_eq}). Now by examining the structure of the RHS in (\ref{eq:khatri_eq}), it can be verified that for $i = 0, 1, 2, \dots, k-1$, worker $i$ will be assigned $q_A$ submatrices from $\bfA$ and $q_B$ submatrices from $\bfB$. And for $i = k, k+1, k+2, \dots, n-1$, any worker $i$ will be assigned $q_A + (i-k) \times(k_A-1)$ submatrices from $\bfA$ and $q_B + (i-k) \times(k_B-1)$ submatrices from $\bfB$. Thus the maximum number of submatrices will be assigned to worker $n-1$, which will have $q_A + (s-1) \times(k_A-1)$ submatrices from $\bfA$ and $q_B + (s-1) \times(k_B-1)$ submatrices from $\bfB$, since $s = n - k$. For the assignment of this worker,
\begin{align*}
C^{\bfA}_{n-1} (D) = \bfU^A_0(D) & + \bfU^A_1(D) \; D^{z(s-1)}  + \dots  + \bfU^A_{k_A-1}(D) \; D^{z(k_A-1)(s-1)} \; ; \; \; \textrm{and} \\
C^{\bfB}_{n-1} (D) = \bfU^B_0(D) & + \bfU^B_1(D) \; D^{s-1} + \dots  + \bfU^B_{k_B-1}(D) \;  D^{(k_B-1)(s-1)} \; .
\end{align*}
It can be verified that $\bfC^{\bfA}_{n-1}(D)$ is a polynomial in $D$ where the exponent of $D$ at any term is an integer multiple of $z$. Since each $\bfU^B_i(D)$ has a degree $q_B - 1$, the degree of $C^{\bfB}_{n-1} (D)$ is $q_B - 1 + (s - 1)(k_B - 1)$, and thus we conclude that
\begin{align}
\label{eq:z}
z \; \geq \; q_B + (s - 1)(k_B - 1).
\end{align} It should be noted that this value of $z$ is large enough for \eqref{lem1eq} to hold.

Next, using an approach similar to (\ref{eq:lower_bd_Delta}), we can derive
 \begin{align*}
q_A \geq \frac{(s-1)(k_A-1)}{k_A(\gamma_A - \frac{1}{k_A})} \; \; \; \textrm{and} \; \; \; q_B \geq \frac{(s-1)(k_B-1)}{k_B(\gamma_B - \frac{1}{k_B})} .
\end{align*}
\begin{example}
\label{matmatexample}

Consider the computation of $\bfA^T \bfB$ over $n=6$ workers and $s=2$ stragglers.
Assume that each worker can store/process $\gamma_A = 5/8$ fraction of matrix $\bfA$ and $\gamma_B =2/3$ fraction of matrix $\bfB$.  We set $k_A = k_B = 2$, so that $q_A = 4$ and $q_B = 3$. By setting $z = q_B+(s-1)(k_B-1) = 4$, we obtain
\begin{align*}
\bfU^A_i(D) & = \sum_{j=0}^{3} \bfA^T_{\langle i,j \rangle} D^{4j}, \; \textrm{for} \; i = 0, 1 ; \\
\textrm{and} \; \; \; \bfU^B_i(D) & = \sum_{j=0}^{2} \bfB_{\langle i,j \rangle} D^{j}, \; \textrm{for} \; i = 0, 1.
\end{align*}
Furthermore,
\begin{align*}
\bfG_A(D) \; & = \; \begin{bmatrix}
1 & 1 & 0 & 0 & 1 & 1 \\
0 & 0 & 1 & 1 & 1 & D^4 \\
\end{bmatrix}  \; \; \; \textrm{and} \\
\bfG_B(D) \; & = \; \begin{bmatrix}
1 & 0 & 1 & 0 & 1 & 1 \\
0 & 1 & 0 & 1 & 1 & D \\
\end{bmatrix}.
\end{align*}
The assignment of jobs to all the workers can be obtained from $[ \, \bfU^A_0(D) \;\; \bfU^A_1(D) \, ] \, \bfG_A(D)$ and $[\bfU^B_0(D) \;\;\;\; \bfU^B_1(D)] \, \bfG_B(D)$. This is shown in Fig. \ref{matmat}.%
\end{example}

\input{matmatex}
\begin{remark}
Our proposed encoding process is very simple and involves only additions at the master node.
\end{remark}

\subsection{Decoding algorithm: Peeling decoder}
\label{onespeel}

Suppose that we obtain results from workers in $\calI \subset \{0, 1, \dots, n-1\}$, with $|\calI| \geq k$. We describe the decoding process below in detail for the matrix-vector case; the discussion is quite similar for the matrix-matrix case.

In the matrix-vector case our unknowns are $\bfu_{il} = \bfA^T_{\langle i,l \rangle}\bfx, i \in [k], l \in [q]$; each of these is a vector of length $r/(kq)$. Let row-vector $\bfz_j$ denote the collection of the $j$-th entries of each of these unknowns, where $j \in [r/(kq)]$. Let the output of the worker nodes corresponding to $\bfz_j$ be denoted by $\bfy_j$. The length of $\bfy_j$ depends on $\calI$.

We assume that the master node obtains results from a subset of the message workers, $\calI_1 \subset \{0, 1, \dots, k-1\}$, so that $|\calI_1| \leq k$. This implies that it can recover $|\calI_1|q$ unknowns directly. Moreover, it obtains results from the parity workers indexed by $\calI_2 \subset \{k, k+1, \dots, n-1\}$, where $|\calI_2|  = k - |\calI_1|$. Thus, it needs to recover the remaining $kq - |\calI_1|q$ unknowns.

The underlying structure of the convolutional code allows for a very simple peeling decoder whereby, at each step, the algorithm is guaranteed to find an equation with only one unknown.
We demonstrate this by means of an example in Appendix \ref{aux:peeling_eg}. Crucially, the scheme can be decoded purely with add/subtract operations and can thus be highly optimized. This algorithm is very fast and has excellent numerical stability ({\it cf.} Fig. \ref{error_18s3_intro}) in experiments.

{\bf Decoding Complexity:}  We consider the worst case where $|\calI_2| = s$. According to the design of this scheme, each of the $kq$ unknowns appears once in every parity worker, and thus the system of equations has at most $kqs$ non-zero entries. Furthermore, in a peeling decoder one variable can be decoded and substituted in the remaining equations at each iteration. Therefore, the time complexity of solving this sparse system is $O(kq s)$. As we solve a total of $r/(kq)$ such systems of equations, the total time taken is $O(rs)$ which is {\it independent} of $q$ and thus does not grow with it; similarly it can be shown that for the matrix-matrix case the time is $O(rws)$.


It should be noted that the matrices $\bfA$ and $\bfB$ are of sizes $t \times r$ and $t \times w$ respectively, thus the computational complexity of computing $\bfA^T \bfB$ is $O(rwt)$. In a distributed system, this job is distributed over $n$ workers with $s$ stragglers, so, on average, the computational complexity of each of the workers is $O\left(\frac{rwt}{k}\right)$, where $k = n - s$. On the other hand, to get the final result, we need to recover $rw$ unknowns, which is the size of $\bfA^T \bfB$. Thus the decoding complexity does not depend on the parameter $t$ which indicates that the decoding time can be often considered negligible in comparison to the worker computation time when $t$ is very large \cite{ramamoorthyDTMag20}. Nevertheless, fast decoding is a desirable feature of any coded computation scheme.


\subsection{Effect of $q$: storage fraction, imbalance in task assignment}

Our presented scheme thus far is provably MDS, efficiently decodable and has excellent numerical stability in experiments. Note that our schemes require lower bounds on the value of $q$ which have an inverse dependence on $\gamma - 1/k$. Thus, if one wants to reduce the imbalance between the task assignments to the message nodes and the parity nodes, then $q$ needs to be chosen large enough.
It turns out that for large values of $q$, the worst case condition number of our scheme can be very large.
We present a theoretical treatment of this phenomenon in the upcoming Section \ref{sec:cond_no_discussion} and discuss techniques for mitigating this effect.

\section{Numerical stability analysis}
\label{sec:cond_no_discussion}

To understand numerical stability, we first introduce a modified encoding scheme and then discuss the matrix representation of the coding ideas described above.

\newcommand{\tG}{\tilde{\bfG}}

\begin{definition}[Randomly scaled generator matrix]
Let $\bfR$ be a $k \times s$ matrix of real numbers. Consider the generator matrix $\bfG(D)$ defined in (\ref{GnkD}). Replace $\bfY_{\bar{b},\bar{a}}(D)$ by $\bfR \circ \bfY_{\bar{b},\bar{a}}(D)$. Here, $\circ$ denotes Hadamard product (.* operation in MATLAB).%
\label{defGD_random}
\end{definition}
Note that if we set $r_{ij}=1$ for all entries of the matrix $\bfR$, we recover the old generator matrix $\bfG(D)$ (the ``All-Ones'' case).

\subsubsection{Understanding the matrix representation}
It is not hard to see that the matrix representation of the transformation induced by the $k \times n$ generator polynomial matrix $\bfG(D)$ from Definition \ref{defGD_random}
can be understood as right multiplying a $kq$-length row vector of input data by the following matrix. An example of this was given in Section \ref{sec:simple_illus}

\begin{definition}[$\tG$: matrix representation of $\bfG(D)$]
\label{defn:tilde_G}
We first define a $q \times (q+h)$ shift matrix that takes a $q$-length row vector and returns a $q+h$-length row vector, where the original vector is shifted to the right by $j$ components. This is the matrix $\tilde{\bfD}^{h;j} \triangleq \begin{bmatrix} \mathbf{0}_{q \times j} & \bfI_{q} & \mathbf{0}_{q \times (h-j)} \end{bmatrix}$.
The $(i, \ell)$-th block matrix of $\tG$ for $\ell=0,1,\dots, k-1$ and $i = 0, 1, \dots, k-1$ is

\[
(\tG)_{i,\ell} =  
\begin{cases}
\bfI_q  \ &\text{ if } \ i=\ell  \\
\bf{0}_{q\times q}  \ &\text{ if } \ \it{i} \neq \it{\ell}
\end{cases}
\vspace{-0.05 in}
\]
and for $\ell=k+j$, $j=0,1,\dots (s-1)$,
\[
(\tG)_{i,\ell}  =  r_{ij} \tilde{\bfD}^{a_j b_{k-1};a_j b_i}.
\]

Thus, $\tG$ is a $kq \times (nq+\delta)$ matrix where
\begin{align}
\label{delta}
\delta \; = \; b_{k-1} \sum\limits_{j=0}^{s-1} \; a_j.
\end{align}

\end{definition}

With the above definition, decoding can be understood as inverting the specific $k \times k$ block submatrix of $\tG$, denoted $\tG_\calI$ where $\calI$ is the set of indices of the $k$ workers that have returned their jobs.

\subsubsection{Quantifying round-off error amplification}
When assuming perfectly noise-free computations, invertibility of the decoding matrix, $\tG_\calI$, is sufficient to guarantee perfect recovery/decoding of the desired matrix-matrix product. However, since all computing devices are finite precision, matrix multiplications will frequently result in bit overflow/underflow and hence round-off errors.
As explained earlier ({\it cf.} Section \ref{relwork}), the decoding process amplifies the round-off error by a factor that can at most be as large as the condition number of the decoding matrix. Thus, the numerical stability of our scheme is quantified by the largest condition number over all block submatrices $\tG_\calI$, i.e., by



\[
\kappa_{worst} \triangleq \max_{\calI \subset [n], |\calI|=k} \kappa( \tG_\calI).
\vspace{-0.15 in}
\]

\subsection{Upper bounding $\kappa_{worst}$}
\label{sec:upper_bd_kappa}
Observe that the matrix $\tG$, and consequently the decoding submatrix $\tG_\calI$ with $|\calI|=k$, has a very specific structure. Because of this, it  is possible to show that the matrix $\tilde{\bfG}_\calI \tilde{\bfG}_\calI^T$ is a $k \times k$ block matrix with Toeplitz blocks of size $q \times q$, see in Appendix \ref{sec:proof_regalia_stuff}. This fact is useful since the asymptotics of $\lambda_{\max}(\tilde{\bfG}_\calI \tilde{\bfG}_\calI^T)$ and $\lambda_{\min}(\tilde{\bfG}_\calI \tilde{\bfG}_\calI^T)$ when $q$ is large have been studied in \cite{gazzah2001asymptotic}. In particular, Theorem $3$ of \cite{gazzah2001asymptotic} shows that using Fourier transform ideas, one can bound the eigenvalues of such matrices by computing the minimum (and maximum) of the smallest (and largest) eigenvalues of a much smaller $k \times k$ matrix that is a function of a scalar parameter $\omega$ which lies in  $[-\pi, \pi]$.

With some abuse of notation, let $\bfG_\calI(e^{\textrm{i} \omega})$ represent the matrix obtained by extracting $\bfG_\calI(D)$ (from $\bfG(D)$ in (\ref{GnkD})) and then substituting $D = e^{\textrm{i} \omega}$ (where $\textrm{i} = \sqrt{-1}$). By adapting the results of \cite{gazzah2001asymptotic} (see Appendix \ref{sec:proof_regalia_stuff} for a detailed description), we have the following theorem.

\begin{theorem}
\label{theorem:regalia_stuff}
For  $\calI \subset \{0, \dots, n-1\}$ such that $|\calI| = k$, we have
\begin{align*}
\lim_{q\to\infty} \lambda_{\min} (\tilde{\bfG}_\calI \tilde{\bfG}_\calI^*) & = \min_{\omega \in [-\pi,\pi]} \lambda_{\min} [(\bfG_\calI(e^{\textrm{i}\omega}))(\bfG_\calI(e^{\textrm{i}\omega}))^*] ; \\
\textrm{and} \; \; \; \; \lim_{q\to\infty} \lambda_{\max} (\tilde{\bfG}_\calI \tilde{\bfG}_\calI^*) & = \max_{\omega \in [-\pi,\pi]} \lambda_{\max} [(\bfG_\calI(e^{\textrm{i} \omega}))(\bfG_\calI(e^{\textrm{i} \omega}))^*].
\end{align*}

Moreover, for any $q$
\begin{align*}
 \lambda_{\max} (\tilde{\bfG}_\calI \tilde{\bfG}_\calI^*) &\leq \max_{\omega \in [-\pi,\pi]} \lambda_{\max} [(\bfG_\calI(e^{\textrm{i} \omega}))(\bfG_\calI(e^{\textrm{i}\omega}))^*] ; \\
\textrm{and} \; \; \; \; \lambda_{\min} (\tilde{\bfG}_\calI \tilde{\bfG}_\calI^*) &\geq  \min_{\omega \in [-\pi,\pi]} \lambda_{\min} [(\bfG_\calI(e^{\textrm{i} \omega}))(\bfG_\calI(e^{\textrm{i}\omega}))^*].
\end{align*}
\end{theorem}

Theorem \ref{theorem:regalia_stuff} shows that we can find an upper bound on the condition number of $\tilde{\bfG}_\calI$ based on a scalar optimization over $\omega \in [-\pi,\pi]$.
When $\bfR$ is chosen to be the all-ones matrix, the characterization of Theorem \ref{theorem:regalia_stuff} allows us to conclude that when $s > 1$, there exist choices of $\calI \subseteq \{0, 1, \dots, n-1\}, |\calI| = k$ such that  $\tilde{\bfG}_\calI \tilde{\bfG}_\calI^*$ has a minimum eigenvalue that will go to zero as $q \rightarrow \infty$. In particular, the corresponding $\bfG_{\calI}(e^{\textrm{i}\omega})$ has repeated columns for $\omega=0$.

\begin{example}
Consider the $(n,k) = (4,2)$ example with $G(D) = \begin{bmatrix}
1 & 0 & 1 & 1 \\
0 & 1 & 1 & D
\end{bmatrix}$. Suppose that $\calI = \{2,3\}$. This implies that
\begin{align*}
\tilde{\bfG}_{\calI}\tilde{\bfG}_{\calI}^T &= \tilde{\bfG}_2\tilde{\bfG}_2^T + \tilde{\bfG}_3\tilde{\bfG}_3^T = \begin{bmatrix}
2\bfI_q & \bfI_q + \bfL\\
\bfI_q + \bfU & 2\bfI_q
\end{bmatrix},
\end{align*}
where $\bfU$ and $\bfL$ are $q \times q$ upper shift and lower shift matrices respectively (see, e.g., \eqref{eq:shift_matrix} in the Appendix).

The corresponding $\bfG_\calI(e^{\textrm{i}\omega}) \bfG_\calI(e^{\textrm{i}\omega})^*$ can be obtained as
\begin{align*}
\bfG_\calI(e^{\textrm{i}\omega}) \bfG_\calI(e^{\textrm{i}\omega})^* = \begin{bmatrix}
                 2 & 1 + e^{-\textrm{i}\omega} \\
                 1 + e^{\textrm{i}\omega} & 2
               \end{bmatrix}
\end{align*}
Using Theorem \ref{theorem:regalia_stuff}, we can conclude therefore that $\lim_{q\to\infty} \lambda_{max} [\calT] = 2$ (achieved at $\omega = \pi$) and $\lim_{q\to\infty} \lambda_{min} [\calT] = 0$ (achieved at $\omega = 0$). This implies therefore that as $q$ becomes larger and larger, the matrix $\tilde{\bfG}_{\calI}$ becomes more and more ill-conditioned, though it is nonsingular for any fixed $q$.
\end{example}

Therefore considering a nontrivial scaling of the parity part with a matrix $\bfR$ is essential for well-conditioned behavior when $q$ is very large.

\subsection{Randomly-weighted convolutional coding}
\label{sec:randomconv}
We now show that choosing the matrix $\bfR$ randomly in Definition \ref{defGD_random} results in better numerical stability than the All-Ones scheme in the regime of large $q$ but requires marginally higher decoding complexity.


The following result shows that the MDS property continues to holds with probability 1 when the entries are chosen i.i.d. from a continuous distribution. The proof is an easy consequence of Theorem \ref{theorem:nonsingular_G} and appears in the Appendix.
\begin{corollary}\label{thm2_cor}
If the entries of the matrix $\bfR$ are chosen i.i.d. from any continuous-valued probability distribution, then, any $k \times k$ submatrix of the generator matrix mentioned in Definition \ref{defGD_random} is non-singular with probability one.

\end{corollary}

We now demonstrate that choosing the matrix $\bfR$ randomly allows us to upper bound the worst case condition number (over the recovery matrices) even when $q \rightarrow \infty$.
In the matrix-vector scenario, Theorem \ref{theorem:regalia_stuff} suggests the following algorithm for choosing $\bfR$. We proceed by randomly choosing $\bfR$. Let $\calI \subset \{0, \dots, n-1\}, |\calI|=k$ and let $\Omega = \{0, \pm \frac{\pi}{N}, \pm \frac{2\pi}{N}, \dots,  \pm \frac{(N-1)\pi}{N}, \pm \pi \}$ for a large positive integer $N$ denote a fine enough grid of the interval $[-\pi, \pi]$. Let $\kappa_{\bfR}$ be defined as
\begin{align*}
\max_{\underset{|\calI| = k}{\calI \subset \{0, \dots, n-1\},}} \sqrt{\left( \frac{\underset{\omega \in \Omega}{\max} \; \; \lambda_{\max} [(\bfG_\calI(e^{\textrm{i}\omega}))(\bfG_\calI(e^{\textrm{i}\omega}))^*]}{\underset{\omega \in \Omega}{\min} \; \; \lambda_{\min} [(\bfG_\calI(e^{\textrm{i}\omega}))(\bfG_\calI(e^{\textrm{i}\omega}))^*]} \right)}.
\end{align*}

Thus, $\kappa_{\bfR}$ indicates the maximum condition number of $\bfG_\calI(e^{\textrm{i}\omega})$ over all $\binom{n}{k}$ choices of $\calI$; this is an upper bound on the maximum condition number of $\tilde{\bfG}_\calI$.
The algorithm repeatedly generates choices of $\bfR$ and retains the choice that has the lowest value of $\kappa_{\bfR}$; this denoted by $\bfR^\star$. The matrix-matrix case is similar, except that we generate two random matrices denoted $\bfR_A$ and $\bfR_B$ and consider the worst case condition number of the appropriate submatrices of (\ref{eq:khatri_eq}) to obtain $\bfR_A^\star$ and $\bfR_B^\star$.
We emphasize that even though the search requires optimizing over $\binom{n}{k}=\binom{n}{s}$ choices of $\calI$, this is a one-time cost for designing the coding scheme for a system with $n$ worker nodes which is resilient to $s=n-k$ stragglers. Furthermore, (i) the search  does not have any dependence on $q$, and (ii) the value of $s$ is typically a small constant, that either does not grow or grows very slowly with $n$. Thus the complexity of the above design, $n^s$, grows as polynomial in $n$. Appendix \ref{sec:more_num_exp} presents some numerical results on the amount of time taken to find a good $\bfR$ matrix.

For systems with $n=12,s=3$ and $n=13,s=3$, we conducted $50$ random trials each to find the corresponding $\bfR^\star$ for the matrix vector multiplication case; the entries were sampled i.i.d. from the uniform distribution on $[-1,1]$. Our algorithm also returns the asymptotic upper bound on $\kappa(\bfR^\star)$. By sweeping over values of $q$, we can also compute the actual worst-case condition number for each particular chosen value of $q$. Fig. \ref{ubn12s3} depicts the upper bound and the actual worst case condition numbers for different $n$ and $s$.

\input{condrandn12s3}

\subsection{Random convolutional coding: decoding algorithm}
In principle, it is possible to use a fast peeling decoder for decoding as done earlier in the all-ones case. Note however that the peeling decoder solves a system of $kq$ equations in $kq$ variables. Thus, it only uses $kq$ columns of the $\tG_{\calI}$ even though $\tG_{\calI}$ is a matrix of size $kq \times (kq + \delta^{'})$ where $\delta^{'}$ is an integer between zero and $\delta$ ({\it cf.} \eqref{delta}), depending on which set of $k$ worker nodes finished their computations (in matrix-vector multiplication).

In particular, the stability of the peeling decoder depends on the condition number of the relevant full rank square submatrix of $\tG_{\calI}$. In general, this condition number is higher than that of $\tG_{\calI}$. In our numerical experiments we have found that for the all-ones case, the worst case condition numbers of both matrices ($\tG_{\calI}$ and full rank square submatrix of $\tG_{\calI}$) are almost the same (see more experimental details in Section \ref{sec:numerical_exp}). This explains the numerically stable behavior of the peeling decoder in the all-ones case.

The situation changes quite a bit when we consider random scaling of the generator matrix. e.g., when the entries of $\bfR$ are i.i.d. random Gaussian, the difference is very large.
In this case, the condition number of the full rank square submatrix of $\tG_{\calI}$ can be very high for certain sets of workers $\calI$ (see in Section \ref{sec:numerical_exp}). But in all cases, $\kappa_{worst}$ over all $\tG_{\calI}$ is significantly smaller than that of the all-ones case. Thus, it is clear that one should use all the columns of $\tG_{\calI}$ for decoding, rather than using only $kq$ equations.


{\bf Decoding Complexity:} Similar to the discussion in Section \ref{onespeel}, we assume that the fastest $k$ workers include the message worker set  $\calI_1$ and the parity worker set $\calI_2$, so that $|\calI_1| + |\calI_2| = k$. We can decode some unknowns directly from the workers in $\calI_1$, and in the worst case, we need to recover the other $sq$ unknowns from the parity workers in $\calI_2$. In this case, one can solve a least square (LS) problem to recover the $sq$ unknowns. This LS problem can be solved in different ways. The most straightforward way would be matrix inversion ($O\left( (sq)^3 \right)$ time) followed by solving $\frac{rw}{kq}$ systems of equations ($O\left( \frac{rw}{kq} (sq)^2 \right)$ time). If $sq \ll r , w$; we can write it as $O\left( \frac{rw}{k} s^2 q \right)$. On the other hand if the value of $q$ is large, then we can use techniques such as conjugate gradient descent to solve the LS problem. This is especially useful when $q$ is large since the underlying system of equations is sparse. Thus, each iteration of conjugate gradient descent can be solved in a fast manner. In particular, if we run it for $T$ iterations to recover these $sq$ unknown blocks, the decoding complexity is $O\left(\frac{rw}{kq} \times sq  \times s \times T\right) = O\left(\frac{rw}{k} s^2 T\right)$. To reach within $\epsilon $ fraction of the solution, the number of iterations scales a $O(\kappa \log(1/\epsilon))$ where $\kappa$ is the condition number of the linear system of equations.

Overall the decoding complexity of the random convolutional code setting is marginally higher than the All-Ones case, depending on which algorithm is used for the LS solution.



\begin{table}[t]
\caption{{\small Comparison of Worst Case Condition Numbers ($\kappa_{worst}$) for Matrix-matrix Multiplication for  $n = 18$ and $s = 3$}}
\label{matmatcondnumber-earlier}

\begin{center}
\begin{small}
\begin{sc}
\begin{tabular}{c c}
\hline
\toprule
Methods & $\kappa_{worst}$\\
 \midrule
Polynomial Code  \cite{yu2017polynomial} &  $4.031 \times 10^7$\\
Ortho-Poly Code  \cite{8849468} & $2.506 \times 10^4$ \\
Random Khatri-Rao Code\cite{8919859} & $5329.3$\\
 Circulant and Rotation Matrix \cite{ramamoorthy2019numerically} & 102 \\
Proposed All-ones Conv Code  & $\mathbf{4417.8}$  \\
Proposed Random Conv Code  & $\mathbf{1829.4}$ \\
\bottomrule
\end{tabular}
\end{sc}
\end{small}
\end{center}
\end{table}%

\input{num_exp}

\section{Conclusions and Future Work}
\label{sec:conclusion}
Most current approaches for coded computation work within the framework of block codes. In this work we presented a convolutional approach to coded matrix computation. Our codes possess simple encoding and decoding algorithms. We demonstrated novel connections between the analysis of numerical stability of our codes and the properties of large Toeplitz matrices. The performance of our codes is better than most of the existing known approaches. It would be interesting to consider other classes of convolutional codes for coded computation and attempt to characterize their properties.


%

\appendix
\subsection{Proof of Theorem \ref{theorem:nonsingular_G} and Corollary \ref{thm2_cor} (MDS property of our codes)}
\label{sec:invertibility_mat}

We begin by a formal description of the field in which the polynomials in the indeterminate $D$ lie. Consider the set of real infinite sequences $\{u_r, u_{r+1}, \dots \}$ for $r \in \mathbb{Z}$ that start at some finite integer index $r$, and continue thereafter. These sequences can be treated as elements of the formal Laurent series \cite{fuja1989cross} in indeterminate $D$ with coefficients from $\mathbb{R}$, i.e., $\bfu(D) = \sum\limits_{i=r}^\infty u_i D^i$. Let us denote the ring of formal Laurent series over $\mathbb{R}$ as $\mathbb{R}((D))$ under the normal addition and multiplication of formal power series. It can be shown \cite{niven1969formal}  that $\mathbb{R}((D))$ forms a field, i.e., each non-zero element in it has a corresponding inverse. Thus, the polynomials $\bfu(D) = \sum_{i=0}^\ell u_i D^i$ that we consider in this work are members of $\mathbb{R}((D))$ and can be added, multiplied and divided to obtain other members of $\mathbb{R}((D))$. The zero element and identity element are precisely the real number $0$ and the real number $1$ within this field.

The proof of Theorem \ref{theorem:nonsingular_G} is an immediate consequence of Lemma \ref{lemma:schur_poly} below since any $k \times k$ submatrix of $G(D)$ is of the form $\bfX(D)$ given in the lemma.

\begin{lemma}
\label{lemma:schur_poly}
Consider a square matrix $\bfX(D)$ such that
\begin{equation*}
\bfX(D) = \begin{bmatrix}
    \left( D^{a_0} \right)^{b_0 } & \left( D^{a_1} \right)^{b_0 } & \dots & \left( D^{a_{v-1}} \right)^{b_0 } \\
	\left( D^{a_0} \right)^{b_1} & \left( D^{a_1} \right)^{b_1}  &  \dots & \left( D^{a_{v-1}} \right)^{b_1} \\
	\threevdots & \threevdots & \threevdots & \threevdots \\
	\left( D^{a_0} \right)^{b_{v-1} } & \left( D^{a_1} \right)^{b_{v-1}} & \dots & \left( D^{a_{v-1}} \right)^{b_{v-1} }
\end{bmatrix}\
\end{equation*}
where $a_i$ and $b_j$ are positive integers for $0 \leq i,j \leq v-1$ such that $0 \leq a_0 < a_1 < \dots < a_{v-1}$ and $0 \leq b_0 < b_1 < \dots < b_{v-1}$. Then $\bfX(D)$ is nonsingular, i.e., its determinant is a non-zero polynomial in $D$. Furthermore, if $\bfR$ is a $v \times v$ matrix with entries chosen i.i.d. from a continuous distribution, then $\bfR \circ \bfX(D)$ (where $\circ$ denotes the Hadamard product) is nonsingular with probability 1.
\end{lemma}

The proof of Lemma \ref{lemma:schur_poly} involves Schur polynomials that are defined next.
\begin{definition}
\label{defn:schur_poly}
Let $\lambda_0 \geq \lambda_1 \geq \dots \lambda_{v-1}$ be non-negative integers and let $\mathbf{\lambda} = (\lambda_0, \dots, \lambda_{v-1})$. Then,
\begin{align}
\mathcal{S}_\lambda(x_0, \dots, x_{v-1}) = \sum_{T}  x_0^{t_0} x_1^{t_1} \dots x_{v-1}^{t_{v-1}} \label{eq:young_tableau}
\end{align}
where the summation is over all semistandard Young tableaux $T$ of shape $\mathbf{\lambda}$ \cite{macdonald_symm}.
\end{definition}
A Young diagram of shape $\mathbf{\lambda}$ consists of a collection of boxes arranged in left-justified rows. The $i$-th row has $\lambda_i$ boxes. A semistandard Young tableau $T$ is obtained by filling the boxes with the integers $0, \dots, v-1$ such that entries are in ascending order from left to right in the rows and in strictly increasing order from top to bottom in the columns. The $t_i$ values in (\ref{eq:young_tableau}) are obtained by counting the occurrences of the number $i$ in tableau $T$.

\begin{proof}
Matrix $\bfX(D)$ can be written upon permuting some rows as $\hat{\bfX}(D)$ which is given by
\begin{equation*}
\hat{\bfX}(D) \; = \; \begin{bmatrix}
    \left( D^{a_0} \right)^{\lambda_0 + v - 1} & \left( D^{a_1} \right)^{\lambda_0 + v - 1} & \dots & \left( D^{a_{v-1}} \right)^{\lambda_0 + v - 1} \\
	\left( D^{a_0} \right)^{\lambda_1 + v - 2} & \left( D^{a_1} \right)^{\lambda_1 + v - 2}  &  \dots & \left( D^{a_{v-1}} \right)^{\lambda_1 + v - 2} \\
	\threevdots & \threevdots & \threevdots & \threevdots \\
	\left( D^{a_0} \right)^{\lambda_{v-1} } & \left( D^{a_1} \right)^{\lambda_{v-1}} & \dots & \left( D^{a_{v-1}} \right)^{\lambda_{v-1} }
\end{bmatrix}\
\end{equation*} where we can assume that $\lambda_0 \geq \lambda_1 \geq \dots \geq \lambda_{v-1}$.
\noindent We need to prove that the determinant of $\hat{\bfX}(D)$ is non-zero. According to \cite{macdonald_symm} (Chapter 1),
\begin{align*}
\det (\hat{\bfX}(D)) \; = \; \det\left(\bfZ( D^{a_0}, D^{a_1}, \dots, D^{a_{v-1}} )\right) \; \times \; \mathcal{S}_{\mathbf{\lambda}} \left( D^{a_0}, D^{a_1}, \dots, D^{a_{v-1}} \right),
\end{align*}
where
\begin{align}
\bfZ( D^{a_0}, \dots, D^{a_{v-1}} )  = \; \begin{bmatrix}
    \left( D^{a_0} \right)^{v - 1} & \left( D^{a_1} \right)^{v - 1} & \dots & \left( D^{a_{v-1}} \right)^{v - 1} \\
	\left( D^{a_0} \right)^{v - 2} & \left( D^{a_1} \right)^{v - 2}  &  \dots & \left( D^{a_{v-1}} \right)^{v - 2} \\
	\threevdots & \threevdots & \threevdots & \threevdots \\
	 D^{a_0}  & D^{a_1}  & \dots &  D^{a_{v-1}}  \\
	1 & 1 & \dots & 1
\end{bmatrix} .
\label{eq:vandermonde_eg}
\end{align}
Note that $\det\left(\bfZ( D^{a_0}, D^{a_1}, \dots, D^{a_{v-1}} )\right) $ is a non-zero polynomial in $D$ as it is a Vandermonde matrix.

Furthermore, based on Definition \ref{defn:schur_poly}, $\mathcal{S}_{\mathbf{\lambda}} \left( D^{a_0}, D^{a_1}, \dots, D^{a_{v-1}} \right)$ consists of the sum of terms of the form $\left( D^{a_0} \right)^{t_0} \; \left( D^{a_1} \right)^{t_1} \; \dots \; \left( D^{a_{v-1}} \right)^{t_{v-1}}$ all of which have positive coefficients. Thus, it follows that $\mathcal{S}_{\mathbf{\lambda}} \left( D^{a_0}, D^{a_1}, \dots, D^{a_{v-1}} \right)$ is not the zero-polynomial.
\end{proof}

\begin{proof}[Proof of Corollary \ref{thm2_cor}]
To see the extension, we note that $\det (\bfR \circ \bfX(D))$ is a polynomial in $D$ whose coefficients in turn are multivariate polynomials in the elements of $\bfR$, i.e., $\{r_{i,j}\}, 0 \leq i,j \leq v-1$. Based on the proof above, it is clear that setting $\bfR$ to be a matrix of all-ones results in a nonsingular matrix. This implies that $\det (\bfR \circ \bfX(D))$ is not identically zero. Next, the elements of $\bfR$ are chosen i.i.d. from a continuous distribution. Therefore the probability that all the coefficients evaluate to zero over the random choice is also zero.
\end{proof}

\input{young}

\begin{example}[Illustration of Lemma 2]
Suppose that $v=3$ and consider the square submatrix,
\begin{align*}
\bfE = \begin{bmatrix}
    D^4 & D^8 & D^{16}\\
    D^2 & D^4 & D^8\\
	D & D^2 & D^4
	\end{bmatrix}\
\end{align*} where $\lambda_0 = 2, \lambda_1 = 1$ and $\lambda_2 = 1$, so $\lambda = (2, 1, 1)$. The determinant of $\bfE$ is given by
\begin{align*}
 \det (\bfE) \; & = \; \mathcal{S}_{\lambda} \left( D, D^2, D^4 \right)\; \times \det
\left( \begin{bmatrix}
    D^2 & D^4 & D^8 \\
    D & D^2 & D^4\\
	1 & 1 & 1
\end{bmatrix} \right) \; \\
 & = \; \mathcal{S}_{\lambda} \left( D, D^2, D^4 \right)\; \times \left[ \left( D - D^2 \right) \left(D^2 - D^4 \right) \left(D - D^4 \right) \right]
\end{align*} The Schur polynomial can be obtained from Fig. \ref{schur} as
\begin{align*}
\mathcal{S}_{\lambda} \left( D, D^2, D^4 \right) =   \left(D^{4}\right)^2 \left(D^{2}\right)^1\left(D\right)^1 + \left(D^{4}\right)^1\left(D^{2}\right)^2\left(D\right)^1 + \left(D^{4}\right)^1\left(D^{2}\right)^1\left(D\right)^2  = \;D^{11} + D^9 + D^8.
\end{align*}
\end{example}

\subsection{Example of peeling decoder}
\label{aux:peeling_eg}
\begin{example}
Consider Example \ref{matmatexample} for matrix-matrix multiplication, as shown in Fig. \ref{matmat} and suppose that workers $W0$ and $W1$ are stragglers. The goal of the master node is to recover all products of the form $\bfA^T_{\langle i_1,j_1 \rangle} \bfB_{\langle i_2,j_2 \rangle}$ for $i_1 \in [2], j_1 \in [4], i_2 \in [2], j_2 \in [3]$, hence we have total $2 \times 4 \times 2 \times 3 = 48$ unknowns. Note that we can directly obtain $4 \times 6 = 24$ unknowns from workers $W2$ and $W3$. So it remains to recover all unknowns of the form $\bfA^T_{\langle 0,j_1 \rangle} \bfB_{\langle i_2,j_2 \rangle}$ for $j_1 \in [4], i_2 \in [2], j_2 \in [3]$ from workers $W4$ and $W5$.

First, we concentrate on the first block product of $W5$, which helps to recover $\bfA^T_{\langle 0,0 \rangle} \bfB_{\langle 0,0 \rangle}$. Following this we examine the first block product of $W4$, which is $\left(\bfA_{\langle 0,0 \rangle} + \bfA_{\langle 1,0 \rangle} \right)^T \left(\bfB_{\langle 0,0 \rangle} + \bfB_{\langle 1,0 \rangle}\right)$; the only unknown here is $\bfA^T_{\langle 0,0 \rangle} \bfB_{\langle 1,0 \rangle}$ which can therefore be decoded. We can keep moving back and forth between $W4$ and $W5$ and it can be verified that we can recover all the block products $\bfA^T_{\langle 0,j_1 \rangle} \bfB_{\langle i_2,j_2 \rangle}$ in a similar fashion.
\end{example}

\newcommand{\tB}{\tilde{\bfB}}



\subsection{Proof of Theorem \ref{theorem:regalia_stuff}}
\label{sec:proof_regalia_stuff}

Let $\bar{b}$ be a vector of length $2q-1$, whose entries are indexed as $\bar{b}_\ell, -(q-1) \leq \ell \leq (q-1)$. A Toeplitz matrix of size $q\times q$, denoted by $\mathrm{Toeplitz(\bar{b})}$ is such that its $(i,j)$-th entry is given by $\bar{b}_{i-j}$ for $i \in [q], j \in [q]$. Thus, it is such that each diagonal is a constant from top-left to bottom-right.

Our proof of Theorem \ref{theorem:regalia_stuff} relies on a result from  \cite{gazzah2001asymptotic}.
Consider a $kq \times kq$ matrix $\tB$ that has Toeplitz blocks of size $q\times q$ with the $(i,j)$-th block specified by the $(2q-1)$-length vector $\bar{b}^{i,j}$. 
To be precise, for $i=0,1,\dots,(k-1), \ j=0,1,\dots,(k-1)$,
\[
(\tB)_{i,j} = \mathrm{Toeplitz}(\bar{b}^{i,j}).
\]

The result in \cite{gazzah2001asymptotic} shows that the minimum and maximum eigenvalues of such a matrix can be bounded by computing the minimum and maximum of the eigenvalues of the following ({\it much smaller}) $k \times k$ Fourier transform (FT) matrix $\bfB (\omega)$ over the frequency parameter $\omega$. The $(i,j)$-the entry of $\bfB (\omega)$ is defined by simply computing the Fourier transform of the corresponding vector $\bar{b}^{i,j}$, i.e.,
\[
\left( \bfB(\omega) \right)_{i,j}
= \sum_{\ell=-(q-1)}^{(q-1)} \bar{b}^{i,j}_\ell e^{-\textrm{i} \omega \ell}.
\]
We can now state the result.

\begin{lemma}[Theorem 3 of \cite{gazzah2001asymptotic}] ~\\

\begin{itemize}[wide, labelwidth=!, labelindent=0pt]
\item[(i)] For all $q$, the eigenvalues of $\tilde{\bfB}$ lie in
\[
\left[\underset{\omega \in [-\pi,\pi]}{\min} \lambda_{\min} \mathbf{B}(\omega) \; \; , \;  \underset{\omega \in [-\pi,\pi]}{\max} \lambda_{\max} \mathbf{B}(\omega) \right].
\]
\item[(ii)] Furthermore,
\begin{align}
\lim_{q\to\infty} \lambda_{\min} \left( \tilde{\bfB} \right) & =  \underset{\omega \in [-\pi, \pi]}{\text{min}}  \lambda_{\min} \left(  \mathbf{B}(\omega) \right) ; \\
\lim_{q\to\infty} \lambda_{\max} \left( \tilde{\bfB} \right) & =  \underset{\omega \in [-\pi, \pi]}{\text{max}}  \lambda_{\max} \left(  \mathbf{B}(\omega) \right).
\end{align}

\end{itemize}
\end{lemma}
In other words, the behavior of the eigenvalues of $\tilde{\bfB}$ which is a $kq \times kq$ matrix can be studied instead by computing the eigenvalues of the $k \times k$ matrix $\mathbf{B}( \omega)$ and finding its minimum and maximum eigenvalues over the range $\omega \in [-\pi,\pi]$.

The next two lemmas below help prove that $\tilde{\bfG}_\calI \tilde{\bfG}_\calI^{T}$ has Toeplitz blocks.
\begin{figure*}[!t]
\begin{align}
\label{gtilde}
\centering
\tilde{\bfG} = \begin{bmatrix}
\bovermat{$k$ block-columns}{\bfI_q & 0 & \dots & 0 &} r_{00}\bovermat{$n-k$ block-columns}{\tilde{\bfD}^{a_0 b_{k-1};a_0 b_0} &    \dots & r_{(s-1)0} \tilde{\bfD}^{a_{s-1}b_{k-1};a_{s-1} b_0} &}\\
0 & \bfI_q & \dots & 0 & r_{10} \tilde{\bfD}^{a_0 b_{k-1};a_0 b_1} &    \dots & r_{1(s-1)}\tilde{\bfD}^{a_{s-1}b_{k-1};a_{s-1} b_1}\\
\vdots & \vdots & \ddots & \vdots & r_{20}\tilde{\bfD}^{a_0 b_{k-1};a_0 b_2} & \dots & r_{2(s-1)}\tilde{\bfD}^{a_{s-1}b_{k-1};a_{s-1} b_2}\\
0      &    0   &    0   & \bfI_q & r_{(k-1)0} \tilde{\bfD}^{a_0 b_{k-1};a_0 b_{k-1}} &  \dots & r_{(k-1)(s-1)} \tilde{\bfD}^{a_{s-1}b_{k-1};a_{s-1} b_{k-1}}
\end{bmatrix} .
\end{align}
\end{figure*}

Let $\bfU$ and $\bfL = \bfU^T$ denote square upper and lower shift matrices respectively, i.e., $\bfU$ is a $q\times q$ matrix such that
\begin{align*}
\bfU_{ij} = \begin{cases}
1 & \text{if~} j = i+1\\
0 & \text{~otherwise}.
\end{cases}
\end{align*}
Thus, for instance if $q=5$, then
\begin{align}
\bfU=\begin{bmatrix}
0 & 1 & 0 & 0 & 0 \\
0 & 0 & 1 & 0 & 0 \\
0 & 0 & 0 & 1 & 0 \\
0 & 0 & 0 & 0 & 1 \\
0 & 0 & 0 & 0 & 0
\end{bmatrix} . \label{eq:shift_matrix}
\end{align}

\begin{lemma}
\label{lemma:toeplitz_structure}
Let $h \geq \max(i,j)$. Then
\begin{align*}
(\tilde{\bfD}^{h;i}) (\tilde{\bfD}^{h;j})^T & = \begin{cases}
\bfU^{i-j} & \text{if~} i > j,\\
 \bfL^{j-i} & \text{if~} i \leq j.
\end{cases},
\end{align*}
Note that the matrices on the RHS above are Toeplitz.
\end{lemma}
\begin{proof}
We only prove the case when $i > j$ as the other part is very similar. The product $(\tilde{\bfD}^{h;i}) (\tilde{\bfD}^{h;j})^T$ can be expressed as
\begin{align*}
    \begin{bmatrix}
    \mathbf{0}_{q \times i} & {\bfI_q} & \mathbf{0}_{q \times (h-i)}
    \end{bmatrix}
\times \begin{bmatrix}
    \mathbf{0}_{j \times q}\\\\ {\bfI_q}\\\\ \mathbf{0}_{(h-j) \times q}
    \end{bmatrix} \;  = \;  \begin{bmatrix} \mathbf{0}_{(q - (i-j)) \times (i-j)} & \bfI_{q-(i-j)} \\
\mathbf{0}_{(i-j)\times (i-j)} &  \mathbf{0}_{(i-j)\times (q -(i-j))}
    \end{bmatrix} \; = \;  \bfU^{i-j}.
\end{align*}
\end{proof}

\begin{lemma}
Let $\tilde{\bfG}_\ell$ denote the $\ell$-th block-column of $\tilde{\bfG}$. For  $\ell=0,1,\dots, k-1$,
\begin{align*}
 [\tilde{\bfG}_\ell \tilde{\bfG}_\ell^T]_{i,j} =
\begin{cases}
 \bfI_q & \text{~if~} i=j=\ell,\\
\mathbf{0} & \text{~otherwise.}
\end{cases}
\end{align*}
For $\ell=k+\tilde{\ell}$, $\tilde{\ell}=0,1,\dots, s-1$,
and for $i \ge j$
\begin{align*}
 [\tilde{\bfG}_\ell \tilde{\bfG}_\ell^T]_{i,j} =
\begin{cases}
r_{i \tilde{\ell}}^2 \bfI_q & \text{~if~} i=j,\\
r_{i \tilde{\ell}}r_{j \tilde{\ell}} \bfU^{a_{\tilde{\ell}}(b_i - b_j)} & \text{~if~} i>j.
\end{cases}
\end{align*}
Since the matrix is symmetric, specifying its entries for $i \ge j$ is sufficient.
\end{lemma}
\begin{proof}
This follows directly by using Lemma \ref{lemma:toeplitz_structure} and the definition of $\tilde{\bfG}_\ell$.
\end{proof}

Furthermore, using the property that the sum of Toeplitz matrices is Toeplitz, we can conclude that for any subset $\calI \subset \{0, \dots, n-1\}$ such that $|\calI| = k$, we have that the matrix $\tilde{\bfG}_\calI \tilde{\bfG}_\calI^{T}$ is a matrix with Toeplitz blocks.

For ease of presentation let $\calI = \calI_1 \cup \calI_2$ where $\calI_1 \subseteq \{0, \dots, k-1\}$, $\calI_2 \subseteq \{k, \dots, n-1\}$ and $\calI_1 \cap \calI_2 = \emptyset$ and $\tilde{\ell} = \ell-k$ . Then, for $0 \leq i, j \leq k-1$  and $i \geq j$ we can express the $(i,j)$-th block of $(\tilde{\bfG}_\calI) (\tilde{\bfG}_\calI)^{T}$ as follows.

\begin{align}
[(\tilde{\bfG}_\calI) (\tilde{\bfG}_\calI)^{T}]_{i,j} = \begin{cases}
(\sum_{\ell \in \calI_2} r_{i \tilde{\ell}}^2 )\bfI_q + \mathds{1}_{i \in \calI_1} \bfI_q, & \text{~if $i=j$,}\\
\sum_{\ell \in \calI_2} r_{i \tilde{\ell}} r_{j \tilde{\ell}}\bfU^{a_{\tilde{\ell}} (b_i - b_j)} & \text{~if $i>j$.}
\end{cases} \label{eq:calT_specification}
\end{align}
where $\mathds{1}$ denotes the indicator function. By symmetry it suffices to specify $[(\tilde{\bfG}_\calI) (\tilde{\bfG}_\calI)^{T}]_{i,j}$ for $i \geq j$. Each of the blocks is of dimension $q \times q$.

\begin{proof}[Proof of Theorem \ref{theorem:regalia_stuff}]
We emphasize that our matrix $[(\tilde{\bfG}_\calI) (\tilde{\bfG}_\calI)^{T}]$ (see  (\ref{eq:calT_specification})) has Toeplitz blocks. Let $\tilde{\bfB} = (\tilde{\bfG}_\calI) (\tilde{\bfG}_\calI)^{T}$. Then we have

\begin{align}
\tilde{\bfB}_{i,j} \; = \;  [(\tilde{\bfG}_\calI) (\tilde{\bfG}_\calI)^{T}]_{i,j} \nonumber = \;  \begin{cases}
(\sum_{\ell \in \calI_2} r_{i \tilde{\ell}}^2 )\bfI_q + \mathds{1}_{i \in \calI_1} \bfI_q, & \text{~if $i=j$,}\\
\sum_{\ell \in \calI_2} r_{i \tilde{\ell}} r_{j \tilde{\ell}}\bfU^{a_{\tilde{\ell}} (b_i - b_j)} & \text{~if $i>j$.}
\end{cases}
\end{align}
where $\tilde{\ell} = \ell -k$. Observe $\bfU^a$ is a matrix with 1's on the $(a+1)$-th diagonal and zeros everywhere else. Thus, $\tilde{\bfB}_{i,j}$ is a Toeplitz matrix with the $( a_{\tilde{\ell}}(b_i - b_j) )$-th diagonal equal to $r_{i \tilde{\ell}} r_{j \tilde{\ell}}$. Therefore, the corresponding sequence $\bar{b}^{i,j}$ for $i > j$ is given by
\begin{align*}
\bar{b}^{i,j}_m = \begin{cases}
r_{i \tilde{\ell}} r_{j \tilde{\ell}} & \text{~if~} m = - a_{\tilde{\ell}} (b_i - b_j), \\
0 & \text{~otherwise.}
\end{cases}
\end{align*}

Thus, following the discussion above, we obtain
\begin{align*}
(\bfB(\omega))_{i,j}
= \sum_{\ell \in \calI_2} r_{i \tilde{\ell}} r_{j \tilde{\ell}} \exp \left( \textrm{i} \omega a_{\tilde{\ell}} (b_i - b_j) \right)
\end{align*}

The expressions above can equivalently be expressed as replacing $D$ with $e^{\textrm{i} \omega}$ and then computing the inner product of $\bfG_\calI(e^{j\omega})(i, :)$ with $(\bfG_\calI(e^{\textrm{i} \omega}) (j, :))^* $.
Therefore, we can compactly represent
\begin{align*}
\bfB(\omega)
&= \bfG_\calI(e^{\textrm{i} \omega}) \bfG_\calI(e^{\textrm{i} \omega})^*.
\end{align*}
This concludes the proof.
\end{proof}

\input{supp_num_exp}
\ifCLASSOPTIONcaptionsoff
  \newpage
\fi

\input{conv_20_onecolumn.bbl}
\end{document}

%% file: straggler.tex
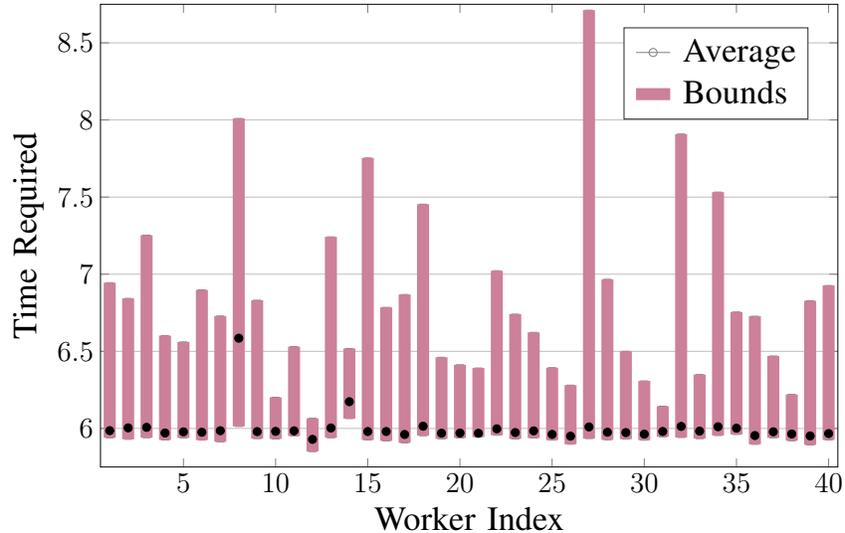
\begin{figure}[t]
\centering
\definecolor{mycolor1}{rgb}{0.80000,0.50000,0.60000}%
\definecolor{mycolor2}{rgb}{0.40000,0.40000,0.80000}%
\centering
\captionsetup{justification=centering}
\resizebox{0.69\linewidth}{!}{
\begin{tikzpicture}

\begin{axis}[%
width=5.1in,
height=3.203in,
at={(2.6in,0.756in)},
scale only axis,
xmin=0.5,
xmax=40.5,
xlabel style={font=\color{white!15!black}, font=\LARGE},
xlabel={Worker Index},
ymin=5.75,
ymax=8.75,
ylabel style={font=\color{white!15!black}, font=\LARGE},
ylabel={Time Required},
ytick={6, 6.5,7,7.5,8,8.5},
xtick={5,10,15,20,25,30,35,40},
tick label style={font=\Large},
axis background/.style={fill=white},
ymajorgrids,
legend style={nodes={scale=1.7}, at={(0.96,0.95)}, legend cell align=left, align=left, draw=white!15!black}
]
\addplot [color=gray, draw=none, mark=o, mark options={solid, gray}]
 plot [error bars/.cd, y dir = both, y explicit]
 table[row sep=crcr, y error plus index=2, y error minus index=3]{%
1	5.98581658840179	0.959149250984192	0.0465415620803835\\
2	6.00359764814377	0.839434473514557	0.0741278243064878\\
3	6.00716423034668	1.24595690727234	0.0672843360900881\\
4	5.97045059204102	0.631181526184082	0.045343589782715\\
5	5.97794937849045	0.582720758914948	0.039126393795013\\
6	5.97529013156891	0.923119711875915	0.0508352041244509\\
7	5.98533693313599	0.743480052947998	0.072195920944214\\
8	6.58506104230881	1.42492607355118	0.572207181453705\\
9	5.97961959838867	0.851500415802002	0.0456934928894039\\
10	5.980953540802	0.220931444168091	0.0493646526336669\\
11	5.98451259851456	0.546078412532807	0.0332155537605283\\
12	5.92907848358154	0.136641597747802	0.0791985034942631\\
13	6.0029022192955	1.23922281503677	0.0634812808036802\\
14	6.17343792200089	0.344864995479583	0.10947188615799\\
15	5.98008903503418	1.77345195770264	0.0548890495300292\\
16	5.98003795862198	0.804897038936615	0.0608618569374082\\
17	5.96088913917542	0.907340965270996	0.0543842697143555\\
18	6.01432245492935	1.43969160795212	0.0612913918495179\\
19	5.96912997484207	0.49222094297409	0.0366241288185121\\
20	5.9702387046814	0.4424183177948	0.0319636058807369\\
21	5.96883571147919	0.422864317893982	0.0247136354446411\\
22	5.99740345239639	1.02502148389816	0.0413964486122129\\
23	5.97340108156204	0.767411935329437	0.0416791558265688\\
24	5.98424435377121	0.63846277475357	0.0466013026237491\\
25	5.96148981332779	0.432592256069183	0.0372907090187073\\
26	5.94988328456879	0.330657658576965	0.0523052835464481\\
27	6.00984258890152	2.70222930669785	0.0753967308998105\\
28	5.9754453420639	0.991303610801697	0.0492685556411745\\
29	5.97305377244949	0.527073304653168	0.0423777890205379\\
30	5.96291655302048	0.345107614994049	0.039221465587616\\
31	5.98049085855484	0.163976018428802	0.0347806763648988\\
32	6.01313875675201	1.89541222095489	0.0715987539291385\\
33	5.98265320062637	0.366771876811981	0.0487893223762512\\
34	6.01112633943558	1.5211336016655	0.056864321231842\\
35	6.00164761781692	0.753265473842621	0.0405026936531065\\
36	5.95397742509842	0.773488442897797	0.0572324585914608\\
37	5.97822363615036	0.491628334522248	0.0391076350212094\\
38	5.96405220508575	0.256363863945007	0.0437552976608275\\
39	5.95139958620071	0.876449325084686	0.0582823967933654\\
40	5.96683704376221	0.959806089401245	0.0413140010833741\\
};
\addlegendentry{Average}

\addplot [color=mycolor1, line width=6.0pt]
  table[row sep=crcr]{%
1	5.93927502632141\\
1	6.94496583938599\\
};

\addplot [color=mycolor1, line width=6.0pt]
  table[row sep=crcr]{%
2	5.92946982383728\\
2	6.84303212165833\\
};

\addplot [color=mycolor1, line width=6.0pt]
  table[row sep=crcr]{%
3	5.93987989425659\\
3	7.25312113761902\\
};

\addplot [color=mycolor1, line width=6.0pt]
  table[row sep=crcr]{%
4	5.9251070022583\\
4	6.6016321182251\\
};

\addplot [color=mycolor1, line width=6.0pt]
  table[row sep=crcr]{%
5	5.93882298469543\\
5	6.5606701374054\\
};

\addplot [color=mycolor1, line width=6.0pt]
  table[row sep=crcr]{%
6	5.92445492744446\\
6	6.89840984344482\\
};

\addplot [color=mycolor1, line width=6.0pt]
  table[row sep=crcr]{%
7	5.91314101219177\\
7	6.72881698608398\\
};

\addplot [color=mycolor1, line width=6.0pt]
  table[row sep=crcr]{%
8	6.0128538608551\\
8	8.00998711585999\\
};

\addplot [color=mycolor1, line width=6.0pt]
  table[row sep=crcr]{%
9	5.93392610549927\\
9	6.83112001419067\\
};

\addplot [color=mycolor1, line width=6.0pt]
  table[row sep=crcr]{%
10	5.93158888816833\\
10	6.20188498497009\\
};

\addplot [color=mycolor1, line width=6.0pt]
  table[row sep=crcr]{%
11	5.95129704475403\\
11	6.53059101104736\\
};

\addplot [color=mycolor1, line width=6.0pt]
  table[row sep=crcr]{%
12	5.84987998008728\\
12	6.06572008132935\\
};

\addplot [color=mycolor1, line width=6.0pt]
  table[row sep=crcr]{%
13	5.93942093849182\\
13	7.24212503433228\\
};

\addplot [color=mycolor1, line width=6.0pt]
  table[row sep=crcr]{%
14	6.0639660358429\\
14	6.51830291748047\\
};

\addplot [color=mycolor1, line width=6.0pt]
  table[row sep=crcr]{%
15	5.92519998550415\\
15	7.75354099273682\\
};

\addplot [color=mycolor1, line width=6.0pt]
  table[row sep=crcr]{%
16	5.91917610168457\\
16	6.78493499755859\\
};

\addplot [color=mycolor1, line width=6.0pt]
  table[row sep=crcr]{%
17	5.90650486946106\\
17	6.86823010444641\\
};

\addplot [color=mycolor1, line width=6.0pt]
  table[row sep=crcr]{%
18	5.95303106307983\\
18	7.45401406288147\\
};

\addplot [color=mycolor1, line width=6.0pt]
  table[row sep=crcr]{%
19	5.93250584602356\\
19	6.46135091781616\\
};

\addplot [color=mycolor1, line width=6.0pt]
  table[row sep=crcr]{%
20	5.93827509880066\\
20	6.4126570224762\\
};

\addplot [color=mycolor1, line width=6.0pt]
  table[row sep=crcr]{%
21	5.94412207603455\\
21	6.39170002937317\\
};

\addplot [color=mycolor1, line width=6.0pt]
  table[row sep=crcr]{%
22	5.95600700378418\\
22	7.02242493629456\\
};

\addplot [color=mycolor1, line width=6.0pt]
  table[row sep=crcr]{%
23	5.93172192573547\\
23	6.74081301689148\\
};

\addplot [color=mycolor1, line width=6.0pt]
  table[row sep=crcr]{%
24	5.93764305114746\\
24	6.62270712852478\\
};

\addplot [color=mycolor1, line width=6.0pt]
  table[row sep=crcr]{%
25	5.92419910430908\\
25	6.39408206939697\\
};

\addplot [color=mycolor1, line width=6.0pt]
  table[row sep=crcr]{%
26	5.89757800102234\\
26	6.28054094314575\\
};

\addplot [color=mycolor1, line width=6.0pt]
  table[row sep=crcr]{%
27	5.93444585800171\\
27	8.71207189559937\\
};

\addplot [color=mycolor1, line width=6.0pt]
  table[row sep=crcr]{%
28	5.92617678642273\\
28	6.9667489528656\\
};

\addplot [color=mycolor1, line width=6.0pt]
  table[row sep=crcr]{%
29	5.93067598342896\\
29	6.50012707710266\\
};

\addplot [color=mycolor1, line width=6.0pt]
  table[row sep=crcr]{%
30	5.92369508743286\\
30	6.30802416801453\\
};

\addplot [color=mycolor1, line width=6.0pt]
  table[row sep=crcr]{%
31	5.94571018218994\\
31	6.14446687698364\\
};

\addplot [color=mycolor1, line width=6.0pt]
  table[row sep=crcr]{%
32	5.94154000282288\\
32	7.90855097770691\\
};

\addplot [color=mycolor1, line width=6.0pt]
  table[row sep=crcr]{%
33	5.93386387825012\\
33	6.34942507743835\\
};

\addplot [color=mycolor1, line width=6.0pt]
  table[row sep=crcr]{%
34	5.95426201820374\\
34	7.53225994110107\\
};

\addplot [color=mycolor1, line width=6.0pt]
  table[row sep=crcr]{%
35	5.96114492416382\\
35	6.75491309165955\\
};

\addplot [color=mycolor1, line width=6.0pt]
  table[row sep=crcr]{%
36	5.89674496650696\\
36	6.72746586799622\\
};

\addplot [color=mycolor1, line width=6.0pt]
  table[row sep=crcr]{%
37	5.93911600112915\\
37	6.46985197067261\\
};

\addplot [color=mycolor1, line width=6.0pt]
  table[row sep=crcr]{%
38	5.92029690742493\\
38	6.22041606903076\\
};

\addplot [color=mycolor1, line width=6.0pt]
  table[row sep=crcr]{%
39	5.89311718940735\\
39	6.8278489112854\\
};

\addplot [color=mycolor1, line width=6.0pt]
  table[row sep=crcr]{%
40	5.92552304267883\\
40	6.92664313316345\\
};

\addplot [color=black, draw=none, mark=*, mark options={solid, fill=mycolor2, black}]
  table[row sep=crcr]{%
1	5.98581658840179\\
2	6.00359764814377\\
3	6.00716423034668\\
4	5.97045059204102\\
5	5.97794937849045\\
6	5.97529013156891\\
7	5.98533693313599\\
8	6.58506104230881\\
9	5.97961959838867\\
10	5.980953540802\\
11	5.98451259851456\\
12	5.92907848358154\\
13	6.0029022192955\\
14	6.17343792200089\\
15	5.98008903503418\\
16	5.98003795862198\\
17	5.96088913917542\\
18	6.01432245492935\\
19	5.96912997484207\\
20	5.9702387046814\\
21	5.96883571147919\\
22	5.99740345239639\\
23	5.97340108156204\\
24	5.98424435377121\\
25	5.96148981332779\\
26	5.94988328456879\\
27	6.00984258890152\\
28	5.9754453420639\\
29	5.97305377244949\\
30	5.96291655302048\\
31	5.98049085855484\\
32	6.01313875675201\\
33	5.98265320062637\\
34	6.01112633943558\\
35	6.00164761781692\\
36	5.95397742509842\\
37	5.97822363615036\\
38	5.96405220508575\\
39	5.95139958620071\\
40	5.96683704376221\\
};
\addlegendentry{Bounds}

\end{axis}

\end{tikzpicture}
}
\caption{\small Variation of worker speeds for the same job over 100 runs across $40$ workers within AWS; the job involves multiplying two random matrices of size $4000 \times 4000$ twice. The average time is shown by the small circle for each worker. The upper and lower edges indicate the maximum and minimum time over the 100 runs. The required time exhibits a wide variation from $5.85$ seconds to $8.71$ seconds.}
\label{strtime}
\end{figure} 

%% file: matvecfig.tex
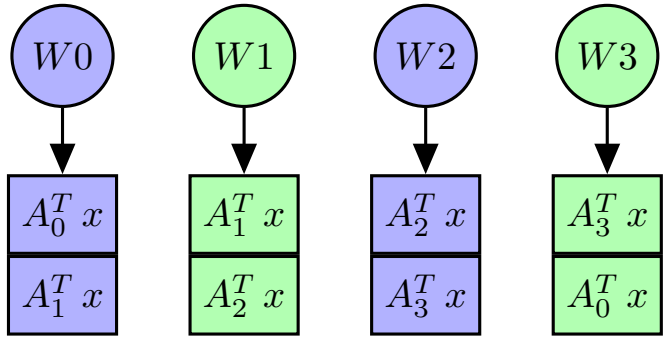
\begin{figure}[t]
\centering
\captionsetup{justification=centering}
\resizebox{0.55\linewidth}{!}{
\begin{tikzpicture}[auto, thick, node distance=2cm, >=triangle 45]

\draw
	node at (0,0)[right=-3mm]{}
	node [sum, fill=blue!30] (blk1) {$W0$}
    node [sum, fill=green!30,right = 0.7cm of blk1] (blk2) {$W1$}
    node [sum, fill=blue!30,right = 0.7cm of blk2] (blk3) {$W2$}
    node [sum, fill=green!30,right = 0.7cm of blk3] (blk4) {$W3$}
    
    node [block, fill=blue!30,below = 0.6 cm of blk1] (blk11) {$A_{0}^T \, x$}
    node [block, fill=blue!30,below = 0.0005 cm of blk11] (blk12) {$A_{1}^T\, x$}
    node [block, fill=green!30,below = 0.6 cm of blk2] (blk21) {$A_{1}^T\, x$}
    node [block, fill=green!30,below = 0.0005 cm of blk21] (blk22) {$A_{2}^T\, x$}
    node [block, fill=blue!30,below = 0.6 cm of blk3] (blk31) {$A_{2}^T\, x$}
    node [block, fill=blue!30,below = 0.0005 cm of blk31] (blk32) {$A_{3}^T\, x$}
    node [block, fill=green!30,below = 0.6 cm of blk4] (blk41) {$A_{3}^T\, x$}
    node [block, fill=green!30,below = 0.0005 cm of blk41] (blk42) {$A_{0}^T\, x$}
    ; 
\draw[->](blk1) -- node{} (blk11);
\draw[->](blk2) -- node{} (blk21);
\draw[->](blk3) -- node{} (blk31);
\draw[->](blk4) -- node{} (blk41);

\end{tikzpicture}
}
\caption{\small Matrix $A$ is divided into four submatrices. Each worker is assigned two of the submatrices and the vector  $x$.}
\label{matvecfig}
\end{figure} 

%% file: matmatfig.tex
\begin{figure}[t]
\centering
\captionsetup{justification=centering}
\resizebox{0.75\linewidth}{!}{
\begin{tikzpicture}[auto, thick, node distance=2cm, >=triangle 45]

\draw
	node at (0,0)[right=-3mm]{}
    	node [sum,fill=blue!30] (blk1) {$W0$}
    node [sum, fill=blue!30,right = 2cm of blk1] (blk2) {$W1$}
    node [sum, fill=blue!30,right = 2cm of blk2] (blk3) {$W2$}
    node [sum, fill=blue!30,below right = 2.6 cm and 0.8 cm of blk1] (blk4) {$W3$}
    node [sum, fill=blue!30,below right = 2.6 cm and 0.8 cm of blk2] (blk5) {$W4$}

    node [block, fill=green!30, minimum width = 2.4 cm, below = 0.4 cm of blk1] (blk11) {$\left( \bfA_{0} +  \textcolor{blue}{1} \, \bfA_{1} \right)$}
    node [block, fill=green!30, minimum width = 2.4 cm,below = 0.0005 cm of blk11] (blk12) {$\left( \bfB_{0} +  \textcolor{blue}{1^2} \, \bfB_{1} \right)$}
    node [block, fill=green!30, minimum width = 2.4 cm,below = 0.4 cm of blk2] (blk21) {$\left( \bfA_{0} +  \textcolor{blue}{2} \, \bfA_{1} \right)$}
    node [block, fill=green!30, minimum width = 2.4 cm,below = 0.0005 cm of blk21] (blk22) {$\left( \bfB_{0} +  \textcolor{blue}{2^2} \, \bfB_{1} \right)$}
    node [block, fill=green!30, minimum width = 2.4 cm,below = 0.4 cm of blk3] (blk31) {$\left( \bfA_{0} +  \textcolor{blue}{3} \, \bfA_{1} \right)$}
    node [block, fill=green!30, minimum width = 2.4 cm,below = 0.0005 cm of blk31] (blk32) {$\left( \bfB_{0} +  \textcolor{blue}{3^2} \, \bfB_{1} \right)$}
	node [block, fill=green!30, minimum width = 2.4 cm,below = 0.4 cm of blk4] (blk41) {$\left( \bfA_{0} +  \textcolor{blue}{4} \, \bfA_{1} \right)$}
	node [block, fill=green!30, minimum width = 2.4 cm,below = 0.0005 cm of blk41] (blk42) {$\left( \bfB_{0} +  \textcolor{blue}{4^2} \, \bfB_{1} \right)$}
    node [block, fill=green!30, minimum width = 2.4 cm,below = 0.4 cm of blk5] (blk51) {$\left( \bfA_{0} +  \textcolor{blue}{5} \, \bfA_{1} \right)$}
    node [block, fill=green!30, minimum width = 2.4 cm,below = 0.0005 cm of blk51] (blk52) {$\left( \bfB_{0} +  \textcolor{blue}{5^2} \, \bfB_{1} \right)$}
    ;
\draw[->](blk1) -- node{} (blk11);
\draw[->](blk2) -- node{} (blk21);
\draw[->](blk3) -- node{} (blk31);
\draw[->](blk4) -- node{} (blk41);
\draw[->](blk5) -- node{} (blk51);

\end{tikzpicture}
}
\caption{\small Matrices $\bfA$ and $\bfB$ are divided into two block-columns each. Each worker is assigned one coded submatrix from $\bfA$ and another coded submatrix from $\bfB$.}
\label{matmatfig}
\end{figure}
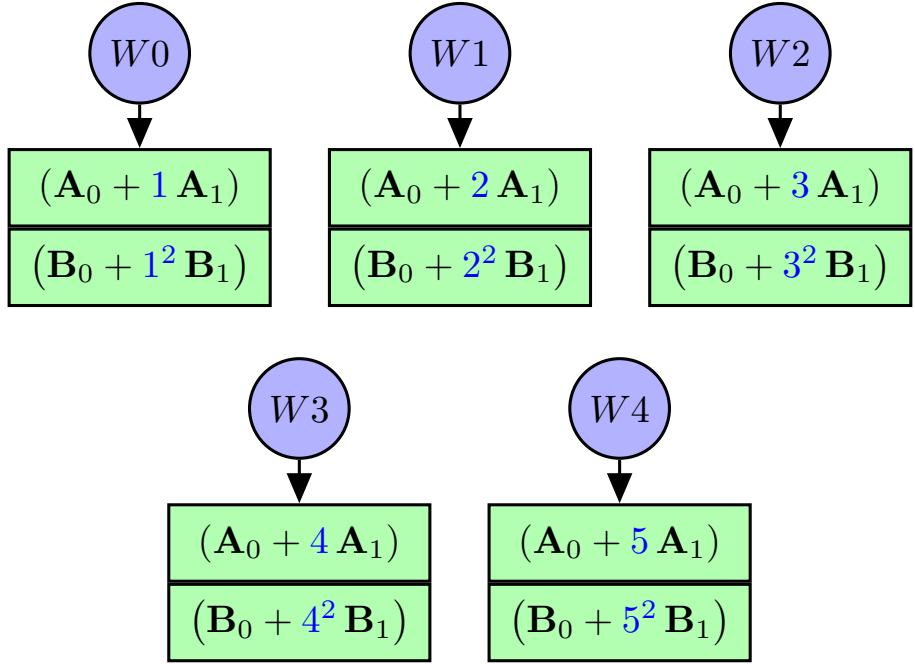 

%% file: err_matmat.tex
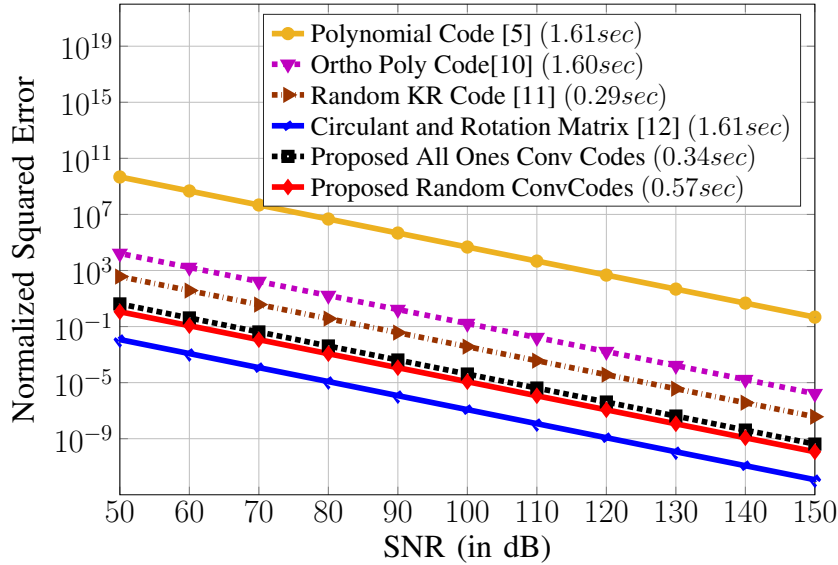
\begin{figure}[t]
\centering
\captionsetup{justification=centering}
\resizebox{0.68\linewidth}{!}{

\definecolor{mycolor6}{rgb}{0.92941,0.69412,0.12549}%
\definecolor{mycolor7}{rgb}{0.74902,0.00000,0.74902}%
\definecolor{mycolor8}{rgb}{0.60000,0.20000,0.00000}%

\begin{tikzpicture}
\begin{axis}[%
width=5.1in,
height=3.603in,
at={(2.6in,0.85in)},
scale only axis,
xmin=50,
xmax=150,
xlabel style={font=\color{white!15!black}, font=\LARGE},
xlabel={SNR (in dB)},
ytick={1e-09,1e-05,1e-01,1e+03,1e+07,1e+11,1e+15,1e+19},
ymode=log,
ymin=1e-13,
ymax=1e+22,
yminorticks=true,
axis background/.style={fill=white},
xmajorgrids,
ymajorgrids,
yminorgrids,
tick label style={font=\LARGE} ,
ylabel style={font=\color{white!15!black}, font=\LARGE},
ylabel={Normalized Squared Error},
axis background/.style={fill=white},
legend style={legend cell align=left, align=left, draw=white!15!black,font = \Large}
]
\addplot [color=mycolor6, line width=3.0pt, mark=o, mark options={solid, mycolor6}]
  table[row sep=crcr]{%
30	469800000000\\
40	47010000000\\
50	4701000000\\
60	470300000\\
70	47010000\\
80	4701000\\
90	470000\\
100	47060\\
110	4697\\
120	470.3\\
130	47.11\\
140	4.692\\
150	0.4715\\
};
\addlegendentry{Polynomial Code \cite{yu2017polynomial} $(1.61 sec)$}

\addplot [color=mycolor7, dashed, line width=3.0pt, mark=triangle, mark options={solid, rotate=180, mycolor7}]
  table[row sep=crcr]{%
30	1652000\\
40	165100\\
50	16520\\
60	1650\\
70	164.9\\
80	16.5\\
90	1.65\\
100	0.1652\\
110	0.0165\\
120	0.00165\\
130	0.0001649\\
140	1.653e-05\\
150	1.646e-06\\
};
\addlegendentry{Ortho Poly Code\cite{8849468} $(1.60 sec)$}

\addplot [color=mycolor8, dashdotted, line width=3.0pt, mark=triangle, mark options={solid, rotate=270, mycolor8}]
  table[row sep=crcr]{%
30	36870\\
40	3684\\
50	368.9\\
60	36.86\\
70	3.684\\
80	0.3687\\
90	0.03686\\
100	0.003685\\
110	0.0003685\\
120	3.686e-05\\
130	3.682e-06\\
140	3.689e-07\\
150	3.683e-08\\
};
\addlegendentry{Random KR Code \cite{8919859} $(0.29 sec)$}

\addplot [color=blue, line width=3.0pt, mark=diamond, mark options={dotted, blue}]
  table[row sep=crcr]{%
50	0.0117699248109080\\
60	0.00118158341272913\\
70	0.000117292925170720 \\
80	1.18994613870255e-05\\
90	1.17847196896449e-06\\
100	1.18574728081043e-07\\
110	1.18053045862444e-08\\
120	1.17727488258126e-09\\
130	1.16901752716387e-10\\
140	1.18488927974027e-11\\
150	1.18770185559698e-12\\
};
\addlegendentry{Circulant and Rotation Matrix \cite{ramamoorthy2019numerically} $(1.61 sec)$}

\addplot [color=black, dotted, line width=3.0pt, mark=square, mark options={solid, black}]
  table[row sep=crcr]{%
30	412.1\\
40	41.24\\
50	4.127\\
60	0.4115\\
70	0.04126\\
80	0.004109\\
90	0.0004115\\
100	4.108e-05\\
110	4.106e-06\\
120	4.125e-07\\
130	4.128e-08\\
140	4.116e-09\\
150	4.131e-10\\
};
\addlegendentry{Proposed All Ones Conv Codes $(0.34 sec)$}

\addplot [color=red, line width=3.0pt, mark=diamond, mark options={solid, red}]
  table[row sep=crcr]{%
30	114.4\\
40	11.53\\
50	1.143\\
60	0.1152\\
70	0.01164\\
80	0.00116\\
90	0.0001166\\
100	1.167e-05\\
110	1.156e-06\\
120	1.152e-07\\
130	1.157e-08\\
140	1.163e-09\\
150	1.169e-10\\
};
\addlegendentry{Proposed Random ConvCodes $(0.57 sec)$}

\end{axis}
\end{tikzpicture}%
}
\vspace{-0.1 in}
\caption{\small Normalized MSE vs. SNR for different coded computation schemes for distributed matrix-matrix multiplication over $n = 18$ workers and $s = 3$ stragglers. The decoding time is reported for the different approaches in parentheses in the legend.}
\label{error_18s3_intro}
\vspace{-0.1 in}
\end{figure} 

%% file: matvecex.tex
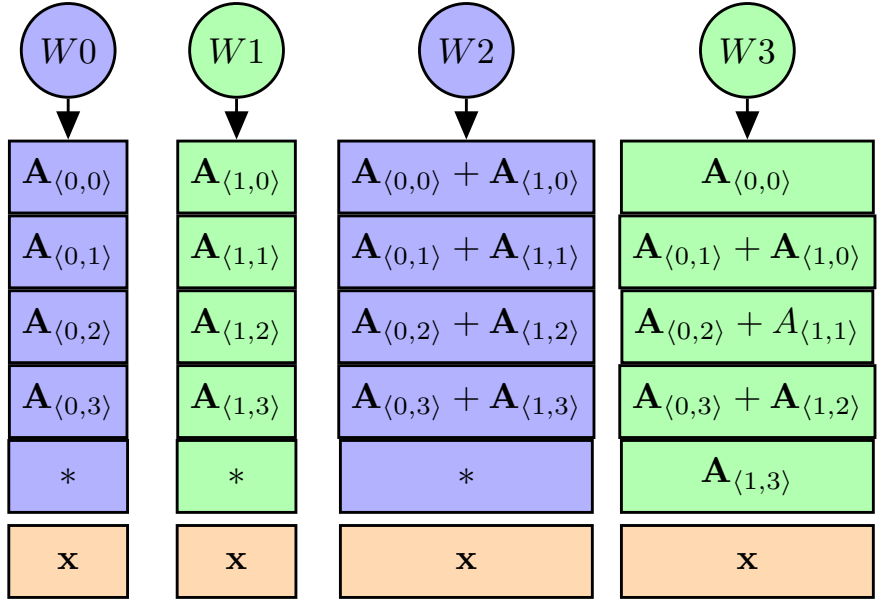
\begin{figure}[t]
\centering
\captionsetup{justification=centering}
\resizebox{0.72\linewidth}{!}{
\begin{tikzpicture}[auto, thick, node distance=2cm, >=triangle 45]

\draw

    node [sum, fill=blue!30] (blk1) {$W0$}
    node [sum, fill=green!30,right = 0.7cm of blk1] (blk2) {$W1$}
    node [sum, fill=blue!30,right = 1.3cm of blk2] (blk3) {$W2$}
    node [sum, fill=green!30,right = 1.8cm of blk3] (blk4) {$W3$}

    node [block, fill=blue!30, minimum width = 2.8em, below = 0.4 cm of blk1] (blk11) {$\bfA_{\langle 0,0\rangle}$}
    node [block, fill=blue!30, minimum width = 2.8em, below = 0.0005 cm of blk11] (blk12) {$\bfA_{\langle 0,1\rangle}$}
    node [block, fill=blue!30, minimum width = 2.8em, below = 0.0005 cm of blk12] (blk13) {$\bfA_{\langle 0,2\rangle}$}
    node [block, fill=blue!30, minimum width = 2.8em, below = 0.0005 cm of blk13] (blk14) {$\bfA_{\langle 0,3\rangle}$}
    node [block, fill=blue!30, minimum width = 3.3em, below = 0.0005 cm of blk14] (blk16) {$*$}
    node [block, fill=orange!30,minimum width = 3.3em,below =  0.1 cm of blk16] (blk17) {$\bfx$}

    node [block, fill=green!30,minimum width = 2.8em,below = 0.4 cm of blk2] (blk21) {$\bfA_{\langle 1,0\rangle}$}
    node [block, fill=green!30,minimum width = 2.8em,below = 0.0005 cm of blk21] (blk22) {$\bfA_{\langle 1,1\rangle}$}
    node [block, fill=green!30,minimum width = 2.8em,below = 0.0005 cm of blk22] (blk23) {$\bfA_{\langle 1,2\rangle}$}
    node [block, fill=green!30,minimum width = 2.8em,below = 0.0005 cm of blk23] (blk24) {$\bfA_{\langle 1,3\rangle}$}
	node [block, fill=green!30,minimum width = 3.3em,below = 0.0005 cm of blk24] (blk26) {$*$}
	node [block, fill=orange!30,minimum width = 3.3em,below =  0.1 cm of blk26] (blk27) {$\bfx$}
	
    node [block, fill=blue!30,below = 0.4 cm of blk3] (blk31) {$ \bfA_{\langle 0,0\rangle} + \bfA_{\langle 1,0\rangle} $}
    node [block, fill=blue!30,below = 0.0005 cm of blk31] (blk32) {$ \bfA_{\langle 0,1\rangle} + \bfA_{\langle 1,1\rangle}  $}
    node [block, fill=blue!30,below = 0.0005 cm of blk32] (blk33) {$ \bfA_{\langle 0,2\rangle} + \bfA_{\langle 1,2\rangle}  $}
    node [block, fill=blue!30,below = 0.0005 cm of blk33] (blk34) {$ \bfA_{\langle 0,3\rangle} + \bfA_{\langle 1,3\rangle}  $}
    node [block, fill=blue!30,below = 0.0005 cm of blk34, minimum width = 7 em] (blk36) {$*$}
	node [block, fill=orange!30,minimum width = 7em, below =  0.1 cm of blk36] (blk37) {$\bfx$}
	
    node [block, fill=green!30,minimum width = 7em, below = 0.4 cm of blk4] (blk41) {$\bfA_{\langle 0,0\rangle} $}
    node [block, fill=green!30,minimum width = 5.6em,below = 0.0005 cm of blk41] (blk42) {$\bfA_{\langle 0,1\rangle} + \bfA_{\langle 1,0\rangle}$}
    node [block, fill=green!30,minimum width = 5.6em,below = 0.0005 cm of blk42] (blk43) {$\bfA_{\langle 0,2\rangle} + A_{\langle 1,1\rangle}$}
    node [block, fill=green!30,minimum width = 5.6em,below = 0.0005 cm of blk43] (blk44) {$ \bfA_{\langle 0,3\rangle} + \bfA_{\langle 1,2\rangle}$}
    node [block, fill=green!30,minimum width = 7em, below = 0.0005 cm of blk44] (blk46) {$ \bfA_{\langle 1,3\rangle} $}
	node [block, fill=orange!30,minimum width = 7em, below =  0.1 cm of blk46] (blk47) {$\bfx$}
;
\draw[->](blk1) -- node{} (blk11);
\draw[->](blk2) -- node{} (blk21);
\draw[->](blk3) -- node{} (blk31);
\draw[->](blk4) -- node{} (blk41);
\end{tikzpicture}
}
\caption{\small Matrix-vector case with $n = 4$ workers and $s = 2$ stragglers, with $\gamma = \frac{5}{8}$.}
\label{matvec}
\vspace{-0.1in}
\end{figure} 

%% file: matmatex.tex
\begin{figure}[t]
\centering
\captionsetup{justification=centering}
\resizebox{0.7\linewidth}{!}{
\begin{tikzpicture}[auto, thick, node distance=2cm, >=triangle 45]
\draw
    node [sum, fill=green!30] (blk1) {$W0$}
    node [sum, fill=blue!30,right = 0.35 cm of blk1] (blk2) {$W1$}
    node [sum, fill=green!30,right = 0.35 cm of blk2] (blk3) {$W2$}
    node [sum, fill=blue!30,right = 0.35 cm of blk3] (blk4) {$W3$}
    node [sum, fill=green!30,right = 1.1 cm of blk4] (blk5) {$W4$}
    node [sum, fill=blue!30,right = 1.7 cm of blk5] (blk6) {$W5$}

    node [block, fill=green!30,below = 0.4 cm of blk1] (blk11) {$\bfA_{\langle 0,0 \rangle}$}
    node [block, fill=green!30,below = 0.0005 cm of blk11] (blk12) {$\bfA_{\langle 0,1 \rangle}$}
    node [block, fill=green!30,below = 0.0005 cm of blk12] (blk13) {$\bfA_{\langle 0,2 \rangle}$}
    node [block, fill=green!30,below = 0.0005 cm of blk13] (blk14) {$\bfA_{\langle 0,3 \rangle}$}
    node [block, fill=green!30,below = 0.0005 cm of blk14,minimum width = 3.25 em] (blk19) {$*$}
    
    node [block, fill=orange!50,below = 0.2 cm of blk19] (blk15) {$\bfB_{\langle 0,0 \rangle}$}
    node [block, fill=orange!50,below = 0.0005 cm of blk15] (blk16) {$\bfB_{\langle 0,1 \rangle}$}
    node [block, fill=orange!50,below = 0.0005 cm of blk16] (blk17) {$\bfB_{\langle 0,2 \rangle}$}
    node [block, fill=orange!50,below = 0.0005 cm of blk17,minimum width = 3.25 em] (blk18) {$*$}

    node [block, fill=blue!30,below = 0.4 cm of blk2] (blk21) {$\bfA_{\langle 0,0 \rangle}$}
    node [block, fill=blue!30,below = 0.0005 cm of blk21] (blk22) {$\bfA_{\langle 0,1 \rangle}$}
    node [block, fill=blue!30,below = 0.0005 cm of blk22] (blk23) {$\bfA_{\langle 0,2 \rangle}$}
    node [block, fill=blue!30,below = 0.0005 cm of blk23] (blk24) {$\bfA_{\langle 0,3 \rangle}$}
    node [block, fill=blue!30,below = 0.0005 cm of blk24,minimum width = 3.25 em] (blk29) {$*$}

    node [block, fill=mycolor2!30,below = 0.2 cm of blk29] (blk25) {$\bfB_{\langle 1,0 \rangle}$}
    node [block, fill=mycolor2!30,below = 0.0005 cm of blk25] (blk26) {$\bfB_{\langle 1,1 \rangle}$}
    node [block, fill=mycolor2!30,below = 0.0005 cm of blk26] (blk27) {$\bfB_{\langle 1,2 \rangle}$}
    node [block, fill=mycolor2!30,below = 0.0005 cm of blk27,minimum width = 3.25 em] (blk28) {$*$}

    node [block, fill=green!30,below = 0.4 cm of blk3] (blk31) {$\bfA_{\langle 1,0 \rangle}$}
    node [block, fill=green!30,below = 0.0005 cm of blk31] (blk32) {$\bfA_{\langle 1,1 \rangle}$}
    node [block, fill=green!30,below = 0.0005 cm of blk32] (blk33) {$\bfA_{\langle 1,2 \rangle}$}
    node [block, fill=green!30,below = 0.0005 cm of blk33] (blk34) {$\bfA_{\langle 1,3 \rangle}$}
    node [block, fill=green!30,below = 0.0005 cm of blk34,minimum width = 3.25 em] (blk39) {$*$}

    node [block, fill=orange!50,below = 0.2 cm of blk39] (blk35) {$\bfB_{\langle 0,0 \rangle}$}
    node [block, fill=orange!50,below = 0.0005 cm of blk35] (blk36) {$\bfB_{\langle 0,1 \rangle}$}
    node [block, fill=orange!50,below = 0.0005 cm of blk36] (blk37) {$\bfB_{\langle 0,3 \rangle}$}
    node [block, fill=orange!50,below = 0.0005 cm of blk37,minimum width = 3.25 em] (blk38) {$*$}

    node [block, fill=blue!30,below = 0.4 cm of blk4] (blk51) {$\bfA_{\langle 1,0 \rangle}$}
    node [block, fill=blue!30,below = 0.0005 cm of blk51] (blk52) {$\bfA_{\langle 1,1 \rangle}$}
    node [block, fill=blue!30,below = 0.0005 cm of blk52] (blk53) {$\bfA_{\langle 1,2 \rangle}$}
    node [block, fill=blue!30,below = 0.0005 cm of blk53] (blk54) {$\bfA_{\langle 1,3 \rangle}$}
    node [block, fill=blue!30,below = 0.0005 cm of blk54,minimum width = 3.25 em] (blk59) {$*$}
    
    node [block, fill=mycolor2!30,below = 0.2 cm of blk59] (blk55) {$\bfB_{\langle 1,0 \rangle}$}
    node [block, fill=mycolor2!30,below = 0.0005 cm of blk55] (blk56) {$\bfB_{\langle 1,1 \rangle}$}
    node [block, fill=mycolor2!30,below = 0.0005 cm of blk56] (blk57) {$\bfB_{\langle 1,2 \rangle}$}
    node [block, fill=mycolor2!30,below = 0.0005 cm of blk57,minimum width = 3.25 em] (blk58) {$*$}

    node [block, fill=green!30,below = 0.4 cm of blk5] (blk41) {$\bfA_{\langle 0,0 \rangle} + \bfA_{\langle 1,0 \rangle}$}
    node [block, fill=green!30,below = 0.0005 cm of blk41] (blk42) {$\bfA_{\langle 0,1 \rangle} + \bfA_{\langle 1,1 \rangle}$}
    node [block, fill=green!30,below = 0.0005 cm of blk42] (blk43) {$\bfA_{\langle 0,2 \rangle} + \bfA_{\langle 1,2 \rangle}$}
    node [block, fill=green!30,below = 0.0005 cm of blk43] (blk44) {$\bfA_{\langle 0,3 \rangle} + \bfA_{\langle 1,3 \rangle}$}
    node [block, fill=green!30,below = 0.0005 cm of blk44,minimum width = 7.05em] (blk49) {$*$}

    node [block, fill=orange!50,minimum width = 4.5em,below = 0.2 cm of blk49] (blk45) {$\bfB_{\langle 0,0 \rangle} + \bfB_{\langle 1,0 \rangle}$}
    node [block, fill=orange!50,minimum width = 4.5em,below = 0.0005 cm of blk45] (blk46) {$\bfB_{\langle 0,1 \rangle} + \bfB_{\langle 1,1 \rangle}$}
    node [block, fill=orange!50,minimum width = 4.5em,below = 0.0005 cm of blk46] (blk47) {$\bfB_{\langle 0,2 \rangle} + \bfB_{\langle 1,2 \rangle}$}
    node [block, fill=orange!50,minimum width = 7em,below = 0.0005 cm of blk47] (blk48) {$*$}

    node [block, fill=blue!30,minimum width = 4.6em, below = 0.4 cm of blk6,minimum width = 7.1em] (blk61) {$\bfA_{\langle 0,0 \rangle}$}
    node [block, fill=blue!30,minimum width = 4.6em,below = 0.0005 cm of blk61] (blk62) {$\bfA_{\langle 0,1 \rangle} + \bfA_{\langle 1,0 \rangle}$}
    node [block, fill=blue!30,minimum width = 4.6em,below = 0.0005 cm of blk62] (blk63) {$\bfA_{\langle 0,2 \rangle} + \bfA_{\langle 1,1 \rangle}$}
    node [block, fill=blue!30,minimum width = 4.6em,below = 0.0005 cm of blk63] (blk64) {$\bfA_{\langle 0,3 \rangle} + \bfA_{\langle 1,2 \rangle}$}
    node [block, fill=blue!30,minimum width = 7.05em, below = 0.0005 cm of blk64] (blk69) {$\bfA_{\langle 1,3 \rangle}$}

    node [block, fill=mycolor2!30,minimum width = 4.6em, below = 0.2 cm of blk69,minimum width = 7.1em] (blk65) {$\bfB_{\langle 0,0 \rangle}$}
    node [block, fill=mycolor2!30,minimum width = 4.6em,below = 0.0005 cm of blk65] (blk66) {$\bfB_{\langle 0,1 \rangle}+ \bfB_{\langle 1,0 \rangle}$}
    node [block, fill=mycolor2!30,minimum width = 4.6em,below = 0.0005 cm of blk66] (blk67) {$\bfB_{\langle 0,2 \rangle} + \bfB_{\langle 1,1 \rangle}$}
    node [block, fill=mycolor2!30,minimum width = 7em, below = 0.0005 cm of blk67] (blk68) {$\bfB_{\langle 1,2 \rangle}$}
    ;

\draw[->](blk1) -- node{} (blk11);
\draw[->](blk2) -- node{} (blk21);
\draw[->](blk3) -- node{} (blk31);
\draw[->](blk4) -- node{} (blk51);
\draw[->](blk5) -- node{} (blk41);
\draw[->](blk6) -- node{} (blk61);

\end{tikzpicture}
}
\caption{\small Matrix-matrix multiplication with $n = 6$ workers and $s = 2$ stragglers with $\gamma_A = \frac{5}{8}$ and $\gamma_B = \frac{2}{3}$.}
\vspace{-0.12in}
\label{matmat}
\end{figure} 

%% file: condrandn12s3.tex
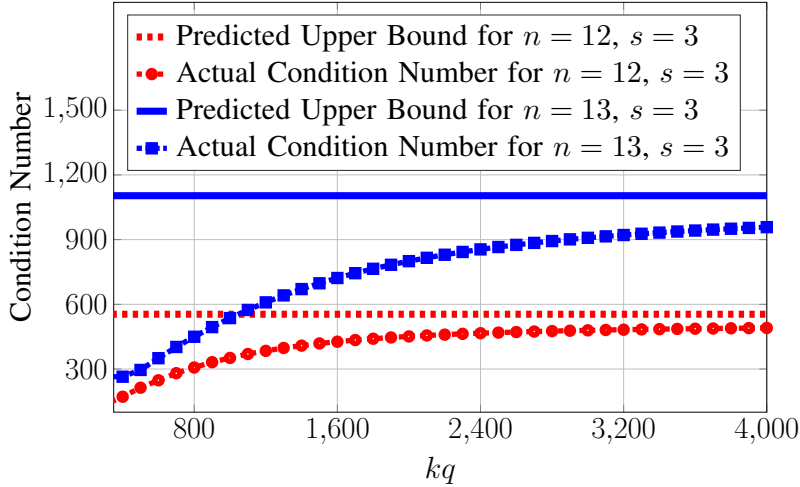
\begin{figure}[t]
\centering
\definecolor{mycolor1}{rgb}{0.00000,0.44700,0.74100}%
\definecolor{mycolor2}{rgb}{0.85000,0.32500,0.09800}%
\definecolor{mycolor3}{rgb}{0.92900,0.69400,0.12500}%
\definecolor{mycolor4}{rgb}{0.49400,0.18400,0.55600}%
\definecolor{mycolor5}{rgb}{0.46600,0.67400,0.18800}%
\captionsetup{justification=centering}
\resizebox{0.65\linewidth}{!}{

\begin{tikzpicture}

\begin{axis}[%
width=5.1in,
height=3.203in,
at={(2.6in,0.756in)},
scale only axis,
xmin=350,
xmax=4000,
xlabel style={font=\color{white!15!black}, font=\LARGE},
xlabel={$k q$},
ymin=100,
ymax=2000,
ylabel style={font=\color{white!15!black}, font= \LARGE},
ylabel={Condition Number},
tick label style={font=\LARGE} ,
axis background/.style={fill=white},
ytick={300,600,900,1200,1500},
xtick={800,1600,2400,3200,4000},
xmajorgrids,
ymajorgrids,
legend style={nodes={scale=1.7}, at={(0.96,0.97)},legend cell align=left, align=left, draw=white!15!black}
]

\addplot [dashed,color=red,, line width = 4.0 pt]
  table[row sep=crcr]{%
100	554.12\\
200	554.12\\
300	554.12\\
400	554.12\\
500	554.12\\
600	554.12\\
700	554.12\\
800	554.12\\
900	554.12\\
1000	554.12\\
1100	554.12\\
1200	554.12\\
1300	554.12\\
1400	554.12\\
1500	554.12\\
1600	554.12\\
1700	554.12\\
1800	554.12\\
1900	554.12\\
2000	554.12\\
2100	554.12\\
2200	554.12\\
2300	554.12\\
2400	554.12\\
2500	554.12\\
2600	554.12\\
2700	554.12\\
2800	554.12\\
2900	554.12\\
3000	554.12\\
3100	554.12\\
3200	554.12\\
3300	554.12\\
3400	554.12\\
3500	554.12\\
3600	554.12\\
3700	554.12\\
3800	554.12\\
3900	554.12\\
4000	554.12\\
};
\addlegendentry{Predicted Upper Bound for $n = 12$, $s = 3$}

\addplot [dashed, mark=o,mark options={solid},color=red, line width = 3.0 pt]
  table[row sep=crcr]{%
100	125.298961123406\\
200	125.298961123407\\
300	135.144643166913\\
400	171.629913371235\\
500	213.419756455019\\
600	247.882702034933\\
700	280.725769067212\\
800	307.105586202819\\
900	331.807092804702\\
1000	351.432628337691\\
1100	369.709123653561\\
1200	384.208969741821\\
1300	397.734115999271\\
1400	408.505568736779\\
1500	418.60702739974\\
1600	426.703698258504\\
1700	434.350552230269\\
1800	440.52569384302\\
1900	446.402288810592\\
2000	451.1843479884\\
2100	455.769659259674\\
2200	459.528711474856\\
2300	463.159078147899\\
2400	466.156088275451\\
2500	469.069931473844\\
2600	471.4909985806\\
2700	473.859435244078\\
2800	475.839026126578\\
2900	477.786547236212\\
3000	479.423169432009\\
3100	481.041603502437\\
3200	482.408417162593\\
3300	483.766421106821\\
3400	484.918485920381\\
3500	486.068060876402\\
3600	487.0473400069\\
3700	488.028357483833\\
3800	488.86721507377\\
3900	489.710600316777\\
4000	490.434276928821\\
};
 \addlegendentry{Actual Condition Number for $n = 12$, $s = 3$}
\addplot [color=blue,, line width = 4.0 pt]
  table[row sep=crcr]{%
100	1103.6\\
200	1103.6\\
300	1103.6\\
400	1103.6\\
500	1103.6\\
600	1103.6\\
700	1103.6\\
800	1103.6\\
900	1103.6\\
1000	1103.6\\
1100	1103.6\\
1200	1103.6\\
1300	1103.6\\
1400	1103.6\\
1500	1103.6\\
1600	1103.6\\
1700	1103.6\\
1800	1103.6\\
1900	1103.6\\
2000	1103.6\\
2100	1103.6\\
2200	1103.6\\
2300	1103.6\\
2400	1103.6\\
2500	1103.6\\
2600	1103.6\\
2700	1103.6\\
2800	1103.6\\
2900	1103.6\\
3000	1103.6\\
3100	1103.6\\
3200	1103.6\\
3300	1103.6\\
3400	1103.6\\
3500	1103.6\\
3600	1103.6\\
3700	1103.6\\
3800	1103.6\\
3900	1103.6\\
4000	1103.6\\
};
\addlegendentry{Predicted Upper Bound for $n = 13$, $s = 3$}

\addplot [dotted,mark=square*,mark options={solid},color=blue, line width = 3.0 pt]
  table[row sep=crcr]{%
100	164.59 \\
200	264.4\\
300	264.4\\
400	264.4\\
500	295.83\\
600	350.36\\
700	402.36\\
800	449.6\\
900	494.41\\
1000	536.18\\
1100	574.01\\
1200	609.13\\
1300	641.44\\
1400	670.65\\
1500	697.5\\
1600	722.06\\
1700	744.31\\
1800	764.67\\
1900	783.29\\
2000	800.2\\
2100	815.69\\
2200	829.88\\
2300	842.82\\
2400	854.7\\
2500	865.62\\
2600	875.63\\
2700	884.85\\
2800	893.35\\
2900	901.18\\
3000	908.43\\
3100	915.13\\
3200	921.34\\
3300	927.11\\
3400	932.47\\
3500	937.45\\
3600	942.09\\
3700	946.43\\
3800	950.47\\
3900	954.25\\
4000	957.80\\
};
 \addlegendentry{Actual Condition Number for $n = 13$, $s = 3$}
\end{axis}

\end{tikzpicture}%
}
\vspace{-0.1 in}
\caption{\small Worst case condition number for the random convolutional code for different $n$ and $s$.}
\label{ubn12s3}
\vspace{-0.25 in}
\end{figure} 

%% file: num_exp.tex
\section{Comparisons and Numerical Experiments}
\label{sec:numerical_exp}


In this section, we discuss the results of the numerical experiments for our proposed approaches and compare our methods with other available methods. 

The polynomial code approach \cite{yu2017polynomial} suffers from the problem that real Vandermonde matrices have condition numbers that are exponential in their size. This in turn implies that for large number of workers (for example, $30$ workers) the condition number of the decoding matrix is so high that the recovered result by the master node is actually useless. 


To avoid this numerical issue, Section VII of \cite{yu2020straggler} remarks that the real computation can be embedded within a large enough finite field of prime order $p$. It turns out that the performance of this scheme is strongly dependent on the entries of $\bfA$ and $\bfB$ and the resultant normalized MSE can be quite bad. These arguments have appeared in \cite{Tang2020Auth}; we present an outline below.

We note that computations in this method are error-free only when each entry of the product matrix $\bfA^T \bfB$ is an integer in $\{0,1,..., p-1\}$. If this requirement is violated, the proposed mod-$p$ computations can return catastrophically wrong answers \cite{Tang2020Auth}. 
This means that the matrices A and B need to be multiplied by a scalar and quantized so that each entry of the resulting matrix is an integer that is within the appropriate range. Suppose that the absolute values of the entries of $\bfA$ and $\bfB$ are upper bounded by $\alpha$; then we need $\alpha^2 t < p$. This is referred to as the dynamic range constraint in \cite{Tang2020Auth}. For instance, with $64$-bit integers (the standard on present day computers), the largest integer is $\approx 10^{19}$. Thus, even if $t < 10^5$, the method can only support $\alpha \leq 10^7$. Thus, the range is rather limited. 

The work of \cite{Tang2020Auth} constructs adversarial $\bfA$ and $\bfB$ integer matrices for this method as follows. Let $p=2147483647$ (note that this is much larger than the publicly available code of \cite{yu2017polynomial} which uses $p=65537$) so that their method can support higher dynamic range. Next let $r=w=t=400$. This implies that $\alpha$ needs to be  $\leq 1000$ by the dynamic range constraint. 
The matrices have the following block decomposition.
\begin{align*}
    \bfA
= \begin{bmatrix}
\bfA_{0,0} & \bfA_{0,1}\\
\bfA_{1,0} & \bfA_{1,1}
\end{bmatrix}, \;\; \text{~and~} \;\;\; \bfB = \begin{bmatrix}
\bfB_{0,0} & \bfB_{0,1}\\
\bfB_{1,0} & \bfB_{1,1}
\end{bmatrix}.
\end{align*}
Each $\bfA_{i,j}$ and $\bfB_{i,j}$ is a matrix of size $200 \times 200$, with entries chosen from the following distributions. $\bfA_{0,0}$, $\bfA_{0,1}$ distributed $\text{Unif}(0, …,9999)$ and $\bfA_{1,0}$, $\bfA_{1,1}$ distributed $\text{Unif}(0, …,9)$. Next, $\bfB_{0,0}$, $\bfB_{0,1}$ distributed $\text{Unif}(0, …,9)$ and $\bfB_{1,0}, \bfB_{1,1}$ distributed $\text{Unif}(0, …,9999)$. In this scenario, the dynamic range constraint requires us to multiply each matrix by $0.1$ and quantize each entry between $0$ and $999$. Note that this implies that $\bfA_{1,0}, \bfA_{1,1}, \bfB_{0,0}, \bfB_{0,1}$ are all quantized into zero submatrices since the entry in these four submatrices is less than $10$. We emphasize that the finite field embedding technique {\it only} recovers the product of these quantized matrices. However, this product is the all-zeros matrix, i.e., the decoded matrix will also be the all-zeros matrix. Therefore, the normalized MSE in this case will be 100 \%.
There are also significant computational issues as discussed in \cite{Tang2020Auth}. We note here that such adversarial can be found even for larger choices of $p$. It is worth noting that the normalized MSE of the other methods do not depend on the actual values of $\bfA$ and $\bfB$.

The work of \cite{8849468} uses orthogonal polynomials and Chebyshev-Vandermonde matrices for the encoding part, which significantly improves the condition number of the decoding matrices compared to \cite{yu2017polynomial} and \cite{dutta2016short}. The work in \cite{8919859} uses random Khatri-Rao product where random coefficients are used for the encoding, which further improves the numerical stability. The recent preprint \cite{ramamoorthy2019numerically} uses circulant and permutation matrices to improve the numerical stability of the polynomial approach. We compare our approaches with these methods with exhaustive numerical experiments which are performed over a cluster in {\tt AWS} (Amazon Web Services). A {\tt t2.2xlarge} machine is used as the master node and {\tt t2.small} machines are used as the slave nodes. Software code for recreating these experiments can be found at \cite{anindyacode}.

{\bf Comparing $\kappa_{worst}$ and MSE for Matrix-matrix case:} For a system with $n = 18$ workers and $s = 3$ stragglers for matrix-matrix multiplication, we set $\gamma_A = \frac{1}{4}$ and $\gamma_B = \frac{2}{5}$ with $k_A = 5$ and $k_B = 3$, so $k = k_A k_B = n - s = 15$. Table \ref{matmatcondnumber-earlier} reports a comparison of the worst-case condition numbers for different approaches in the literature. It can be observed that the work of \cite{yu2017polynomial} and \cite{8849468} have much higher condition numbers than our proposed schemes (All-ones and Random). Both our approaches are also better than the work of \cite{8919859} in terms of worst case condition number ($\kappa_{worst}$) values. 
We point out that the methods in \cite{8849395} and \cite{mallick2018rateless} are developed for matrix-vector multiplication, so those are not applicable for this comparison.

In our next experiment we compare the mean-squared error (MSE) of the different matrix-matrix multiplication methods for their respective worst case scenarios when $n=18$ and $s=3$. For matrix-matrix case, we define MSE as
\begin{align*}
    \textrm{MSE } = \frac{||\bfA^T \bfB - \widehat{\bfA^T \bfB}||^2}{||\bfA^T \bfB||^2} \times 100 \%
\end{align*} where $\widehat{\bfA^T \bfB}$ is the recovered result and $\bfA^T \bfB$ is the actual result. Here, the matrices $\bfA$ and $\bfB$ are of size $15,000 \times 10080$ and $15,000 \times 12000$ respectively. We simulate errors in the worker node computations by adding white Gaussian noise to the calculated submatrix products obtained from the worker nodes and sweeping the range of SNRs. 
The results appear in Fig. \ref{error_18s3_intro} (for additive Gaussian noise) and Fig. \ref{error_18s3_roundoff} (for round-off errors). In Fig. \ref{error_18s3_intro} we observe that even at $SNR= 70 dB$, our approach is around $9$, $4$ and $2$ orders of magnitude better than \cite{yu2017polynomial}, \cite{8849468} and \cite{8919859}.  The corresponding decoding time is also reported in the legend which shows that the decoding time for our approaches compare quite well with other approaches. The behavior of the curves in Fig.  \ref{error_18s3_roundoff} is similar in nature.

\input{err_round_off_matmat}


{\bf Comparing $\kappa_{worst}$ and MSE for Matrix-vector case:} We carry out an experiment to compare the worst case condition number of the decoding matrix for different approaches for matrix-vector multiplication. Table \ref{matveccondnumber} shows the worst case condition number for a scenario with $n = 30$ workers, with $s = 2$ stragglers where each worker node can store $\gamma_A = \frac{1}{25}$ fraction of matrix $\bfA$. From the table, it is clear that the approaches in \cite{yu2017polynomial} and \cite{8849395} provide much larger condition numbers in comparison to the others. From the table, we can also see that our proposed     approaches provide lower condition numbers than the approaches \cite{8849468} and \cite{8919859}.

\begin{table}[t]
\caption{{\small Comparison of $\kappa_{worst}$ for Matrix-vector Multiplication for $n = 30$ and $s = 2$ with $\gamma = \frac{1}{25}$}}
\label{matveccondnumber}
\begin{center}
\begin{small}
\begin{sc}
\begin{tabular}{c c}
\hline
\toprule
Methods & $\kappa_{worst}$ \\ 
 \midrule
 Polynomial Code \cite{yu2017polynomial} & $2.293 \times 10^{13}$ \\
 Convolutional Code \cite{8849395}  & $5.124 \times 10^4$  \\ 
 Ortho-Poly Code \cite{8849468} & $7902.6$\\
 Random KR Code \cite{8919859}  & $3642.7$  \\ 
 Circulant and Rotation Matrix \cite{ramamoorthy2019numerically} & 52 \\
 Proposed All-ones Conv. Code  & $\mathbf{2868.3}$  \\ 
 Proposed Rand Conv. Code  & $\mathbf{1374.6}$\\
\bottomrule
\end{tabular}
\end{sc}
\end{small}
\end{center}
\end{table}%

In our next experiment we compare the normalized MSE of the different methods for their respective worst case scenarios. For matrix-vector case, we define MSE as
\begin{align*}
    \textrm{MSE } = \frac{||\bfA^T \bfx - \widehat{\bfA^T \bfx}||^2}{||\bfA^T \bfx||^2} \times 100 \%
\end{align*} where $\widehat{\bfA^T \bfx}$ is the recovered result and $\bfA^T \bfx$ is the actual result. We consider the same scenario with $n = 30$ and $s = 2$ where we have matrix $\bfA$ of size $30,000 \times 31,500$ and a vector $\bfx$ of length $30,000$. We want to compute the product $\bfA^T \bfx$. Fig. \ref{error_30s2} shows the normalized MSE of the different approaches for different SNR. From the figure we can see that our proposed approaches perform significantly better than all other schemes except the scheme of \cite{ramamoorthy2019numerically}. 
This supports our condition number results in Table \ref{matveccondnumber}. For example, at $SNR = 60dB$, the approach in \cite{8919859} provides around $1.6 \%$ error whereas our all-ones and random convolutional code approaches provide only $0.5 \%$ and $0.2 \%$ error, respectively, for the worst case. 

\input{err_matvec}

\begin{table}[t]
\caption{{\small Comparison of our proposed methods. $n = 11, k_A = k_B = 3$ and $\bfA$ and $\bfB$ have size $10000 \times 12600$.}}
\label{convcomp}
\begin{center}
\begin{small}
\begin{sc}
\begin{tabular}{c c c c c}
\hline
\toprule
Metrics & Methods & $\gamma = \frac{2}{5}$ & $\gamma = \frac{5}{14}$ & $\gamma = \frac{7}{20}$ \\ 
 \midrule
Decoding  & All ones & $0.35 s$ & $0.36 s$ & $0.39 s$ \\
\cline{2-5}
Time & Random & $0.39 s$ & $1.16 s$ & $2.89 s$ \\
 \midrule
$\kappa_{worst}$   & All ones & $95.2$ & $275.9$ & $395.6$ \\
\cline{2-5}
for $\tG_{\calI}$ & Random & $76.9$ & $112.2$ & $117.5$ \\
\midrule
$\kappa_{worst}$ for & All ones & $96.5$ & $277.9$ & $397.8$\\
\cline{2-5}
Sqr. Submat. & \multirow{2}{1 cm}{Random} & $7.46$ & $9.64 $  & $1.11$\\
of $\tG_{\calI}$ &  & $\times 10^6$ & $\times 10^{17}$  & $10^{28}$\\

\bottomrule
\end{tabular}
\end{sc}
\end{small}
\end{center}
\end{table}%

{\bf Comparing \cite{ramamoorthy2019numerically} and our approach: }
It can be observed that the recent preprint of \cite{ramamoorthy2019numerically} has the best $\kappa_{worst}$ and MSE numbers for both the matrix-matrix and matrix-vector scenarios. However, our work has much simpler encoding (additions/subtractions in the All-Ones case) and decoding (peeling decoder) than their method. Our work is also the first to propose a convolutional coding strategy for this problem. 

{\bf Comparing \cite{8919859} and our approach} The Random KR approach can be considered as specific instance of our random scaling method where the scaling is applied to a trivial all-ones parity matrix, instead of a carefully designed $\bfY_{\bar{b},\bar{a}}(D)$. As both approaches are random and pick the best choices, we conducted an experiment where we ran 100 trials for both methods (with $n=20$ and $s=3, 4,5$) and picked the respective best choices (see Fig. \ref{condcomps345} for the corresponding worst case condition numbers). It is clear that the structure imposed in our construction definitely improves the condition number as compared to the work of \cite{8919859}.

\input{cond_rand}

{\bf Comparing our All-ones and random approaches: }
Recall that for our methods $q_A$ and $q_B$ increase when $\gamma_A- 1/k_A$ and $\gamma_B - 1/k_B$ become smaller ({\it cf.} Sections \ref{sec:matvec_section} and \ref{sec:matmat_section}). Table \ref{convcomp}, shows a comparison of our proposed approaches in terms of decoding time and worst case condition number for three different values of $\gamma = \gamma_A=\gamma_B$. The following inferences can be drawn.

\begin{itemize}
\item The decoding time remains more or less constant for the all-ones case, whereas it can increase with decreasing $\gamma$ because of solving LS problem for the random case.
\item The worst case condition number for the all-ones case continues to increase with decreasing $\gamma$, whereas it saturates for the random case.
\item For all-ones case, the worst case condition numbers of both matrices ($\tG_{\calI}$ and full rank square submatrix of $\tG_{\calI}$) are almost the same for different $\gamma$. However, if the entries of $\bfR$ are random Gaussian, then the difference between these two condition numbers is very large.
\end{itemize}

%% file: err_round_off_matmat.tex
\begin{figure}[t]
\centering
\captionsetup{justification=centering}
\resizebox{0.68\linewidth}{!}{

\definecolor{mycolor6}{rgb}{0.92941,0.69412,0.12549}%
\definecolor{mycolor7}{rgb}{0.74902,0.00000,0.74902}%
\definecolor{mycolor8}{rgb}{0.60000,0.20000,0.00000}%

\begin{tikzpicture}
\begin{axis}[%
width=5.1in,
height=3.603in,
at={(2.6in,0.85in)},
scale only axis,
xmin=0,
xmax=10,
xlabel style={font=\color{white!15!black}, font=\LARGE},
xlabel={No. of decimal points (precision)},
ytick={1e-25,1e-20,1e-15,1e-10,1e-05,1e+00,1e+05,1e+10,1e+15},
ymode=log,
ymin=1e-27,
ymax=1e+18,
yminorticks=true,
axis background/.style={fill=white},
xmajorgrids,
ymajorgrids,
yminorgrids,
tick label style={font=\LARGE} ,
ylabel style={font=\color{white!15!black}, font=\LARGE},
ylabel={Normalized Squared Error},
axis background/.style={fill=white},
legend style={legend cell align=left, align=left, draw=white!15!black,font = \Large}
]
\addplot [color=mycolor6, line width=3.0pt, mark=o, mark options={solid, mycolor6}]
  table[row sep=crcr]{%
0	1.2165e+06\\
1	1.2193e+04\\
2	1.2197e+02\\
3	1.2198e+00\\
4	1.2087e-02\\
5	1.2170e-04\\
6	1.2131e-06\\
7	1.2137e-08\\
8	1.2204e-10\\
9	1.2121e-12\\
10	1.2163e-14\\
};

\addlegendentry{Polynomial Code \cite{yu2017polynomial} $(1.61 sec)$}

\addplot [color=mycolor7, dashed, line width=3.0pt, mark=triangle, mark options={solid, rotate=180, mycolor7}]
  table[row sep=crcr]{%
0	1.397885\\
1	1.373448e-02\\
2	1.383753e-04\\
3	1.373895e-06\\
4	1.391267e-08\\
5	1.389363e-10\\
6	1.375464e-12\\
7	1.386957e-14\\
8	1.386143e-16\\
9	1.377376e-18\\
10	1.373654e-20\\
};
\addlegendentry{Ortho Poly Code\cite{8849468} $(1.60 sec)$}

\addplot [color=mycolor8, dashdotted, line width=3.0pt, mark=triangle, mark options={solid, rotate=270, mycolor8}]
  table[row sep=crcr]{%
0	5.5810e-02\\
1	5.5507e-04\\
2	5.5885e-06\\
3	5.5844e-08\\
4	5.5810e-10\\
5	5.5795e-12\\
6	5.6035e-14\\
7	5.5746e-16\\
8	5.5612e-18\\
9	5.5935e-20\\
10	5.6841e-22\\
};
\addlegendentry{Random KR Code \cite{8919859} $(0.29 sec)$}


\addplot [color=black, dotted, line width=3.0pt, mark=square, mark options={solid, black}]
  table[row sep=crcr]{%
0	9.4461e-05\\
1	9.4746e-07\\
2	9.4004e-09\\
3	9.4636e-11\\
4	9.4472e-13\\
5	9.4443e-15\\
6	9.3930e-17\\
7	9.5131e-19\\
8	9.5191e-21\\
9	9.5338e-23\\
10	9.5436e-25\\
};
\addlegendentry{Proposed All Ones Conv Codes $(0.34 sec)$}

\addplot [color=red, line width=3.0pt, mark=diamond, mark options={solid, red}]
  table[row sep=crcr]{%
0	1.7801e-05\\
1	1.7937e-07\\
2	1.8764e-09\\
3	1.8003e-11\\
4	1.8204e-13\\
5	1.9860e-15\\
6	1.9283e-17\\
7	1.7129e-19\\
8	1.7927e-21\\
9	1.6454e-23\\
10	1.9223e-25\\
};
\addlegendentry{Proposed Random ConvCodes $(0.57 sec)$}

\end{axis}
\end{tikzpicture}%
}
\vspace{-0.1 in}
\caption{\small Normalized MSE vs. number of decimal points of precision for different coded computation schemes for distributed matrix-matrix multiplication over $n = 18$ workers and $s = 3$ stragglers. The decoding time is reported for the different approaches in parentheses in the legend.}
\label{error_18s3_roundoff}
\vspace{-0.1 in}
\end{figure}
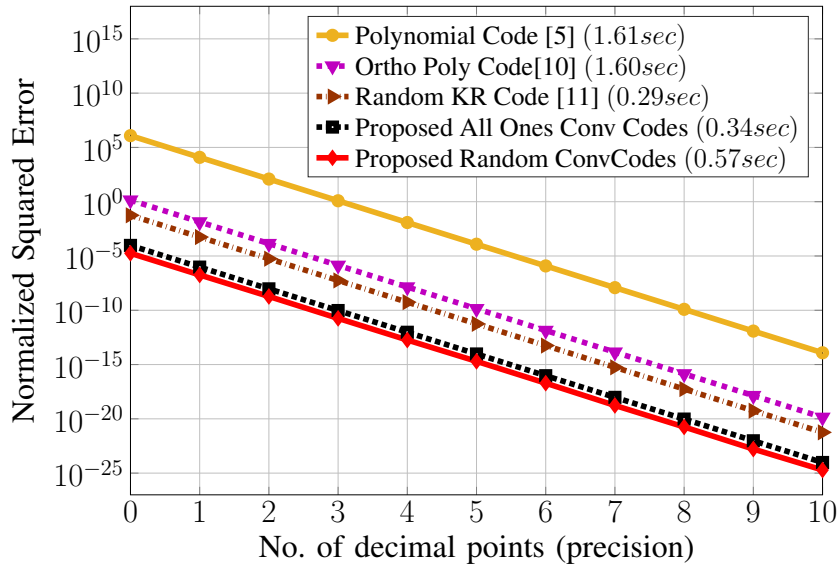 

%% file: err_matvec.tex
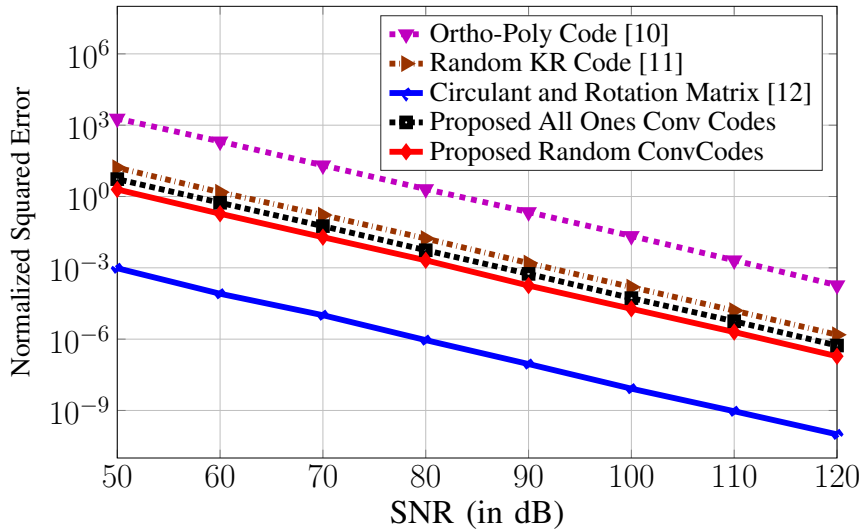
\begin{figure}[t]
\centering
\captionsetup{justification=centering}
\resizebox{0.7\linewidth}{!}{

\definecolor{mycolor6}{rgb}{0.92941,0.69412,0.12549}%
\definecolor{mycolor7}{rgb}{0.74902,0.00000,0.74902}%
\definecolor{mycolor8}{rgb}{0.60000,0.20000,0.00000}%

\begin{tikzpicture}
\begin{axis}[%
width=5.1in,
height=3.203in,
at={(2.6in,0.85in)},
scale only axis,
xmin=50,
xmax=120,
xlabel style={font=\color{white!15!black}, font=\LARGE},
xlabel={SNR (in dB)},
ymode=log,
ymin=1e-11,
ymax=1e+8,
ytick={1e-09,1e-06,1e-03,1e+00,1e+03,1e+06},
xtick={50,60,70,80,90,100,110,120},
tick label style={font=\LARGE} ,
ylabel style={font=\color{white!15!black}, font=\Large},
ylabel={Normalized Squared Error},
axis background/.style={fill=white},
xmajorgrids={true},
ymajorgrids={true},
yminorgrids,
legend style={legend cell align=left, align=left, draw=white!15!black,font = \Large}
]

\addplot [color=mycolor7, dashed, line width=3.0pt, mark=triangle, mark options={solid, rotate=180, mycolor7}]
  table[row sep=crcr]{%
30	209500\\
40	20140\\
50	1898\\
60	207.6\\
70	20.49\\
80	2.052\\
90	0.2213\\
100	0.02184\\
110	0.002072\\
120	0.0001849\\
};
\addlegendentry{Ortho-Poly Code \cite{8849468}}

\addplot [color=mycolor8, dashdotted, line width=3.0pt, mark=triangle, mark options={solid, rotate=270, mycolor8}]
  table[row sep=crcr]{%
30	1534\\
40	150.2\\
50	16.25\\
60	1.516\\
70	0.1657\\
80	0.01705\\
90	0.001586\\
100	0.0001544\\
110	1.587e-05\\
120	1.527e-06\\
};
\addlegendentry{Random KR Code \cite{8919859}}

\addplot [color=blue, line width=3.0pt, mark=diamond, mark options={dotted, blue}]
  table[row sep=crcr]{%
50	0.000959907014771309\\
60	8.11043100141409e-05\\
70	1.01560695764101e-05\\
80	9.09374125174745e-07\\
90	8.95853467134398e-08\\
100	8.24358677158796e-09\\
110	9.27312955214944e-10\\
120	9.88299832399861e-11\\
};
\addlegendentry{Circulant and Rotation Matrix \cite{ramamoorthy2019numerically}}

\addplot [color=black, dotted, line width=3.0pt, mark=square, mark options={solid, black}]
  table[row sep=crcr]{%
30	567.1\\
40	55.21\\
50	5.386\\
60	0.5463\\
70	0.05687\\
80	0.005498\\
90	0.0005566\\
100	5.313e-05\\
110	5.648e-06\\
120	5.391e-07\\
};
\addlegendentry{Proposed All Ones Conv Codes}

\addplot [color=red, line width=3.0pt, mark=diamond, mark options={solid, red}]
  table[row sep=crcr]{%
30	192.1\\
40	20.13\\
50	1.974\\
60	0.1887\\
70	0.01946\\
80	0.002072\\
90	0.0001791\\
100	1.891e-05\\
110	1.989e-06\\
120	1.908e-07\\
};
\addlegendentry{Proposed Random ConvCodes}

\end{axis}
\end{tikzpicture}%
}
\caption{\small Normalized MSE vs SNR plot for matrix-vector multiplication for $n = 30$ and $s = 2$.}
\label{error_30s2}
\end{figure} 

%% file: cond_rand.tex
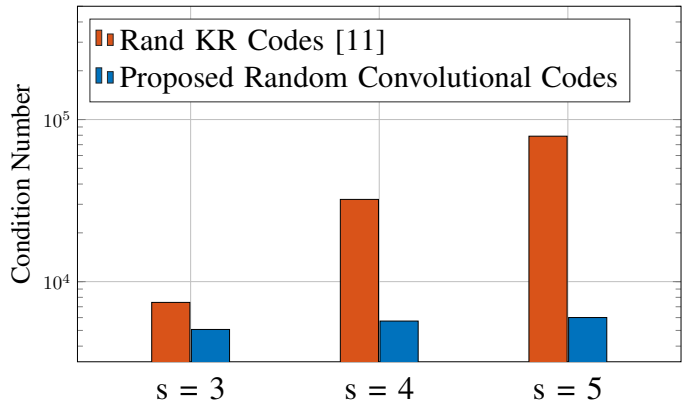
\begin{figure}[t]
\centering
\captionsetup{justification=centering}
\resizebox{0.55\linewidth}{!}{
\begin{tikzpicture}
\begin{axis}[
width=5in,
height=3.203in,
at={(2.6in,0.852in)},
major x tick style = transparent,
ybar=2*\pgflinewidth,
bar width=20pt,
ymajorgrids,
xmajorgrids,
xlabel style={font=\color{white!15!black}, font = \Large},
ylabel style={font=\color{white!15!black}, font = \Large},
ylabel={Condition Number},
ymode=log,
symbolic x coords={{\LARGE s = 3},{\LARGE s = 4},{\LARGE s = 5}},
xtick = data,
scaled y ticks = false,
enlarge x limits= 0.3,
ymin=0,
ymax=500000,
legend cell align=left,
legend style={at={(0.02,0.73)}, nodes={scale=1.6}, anchor=south west, legend cell align=left, align=left, draw=white!15!black}
    ]
    \addplot[style={fill=mycolor2,mark=none}]
            coordinates {({\LARGE s = 3}, 7446.4) ({\LARGE s = 4},32164) ({\LARGE s = 5}, 78974)};
\addlegendentry{Rand KR Codes \cite{8919859}}
    \addplot[style={fill=mycolor1,mark=none}]
             coordinates {({\LARGE s = 3},5070.3) ({\LARGE s = 4},5712.2) ({\LARGE s = 5},6005.5)};
             \addlegendentry{Proposed Random Convolutional Codes}
    \end{axis}

\end{tikzpicture}%
}
\caption{Comparison of $\kappa_{worst}$ for matrix-vector multiplication between the method in \cite{8919859} and our proposed random convolutional code approach for $n = 20$ with $s = 3, 4$ and $5$. To find $\kappa_{worst}$, the proposed method used $\gamma = \frac{1}{15}, \frac{1}{14}, \frac{1}{13}$ for $k = 17, 16, 15$, respectively.}
\label{condcomps345}
\end{figure}

%% file: young.tex
\begin{figure}[t]
\centering
\captionsetup{justification=centering}
\resizebox{0.65\linewidth}{!}{
\begin{tikzpicture}[auto, thick, node distance=2cm, >=triangle 45]
\draw
	node at (0,0)[]{}
	node [block] (block0){$0$}
    node [block, below = 0.0005 cm of block0] (block1) {$1$}
    node [block, below = 0.0005 cm of block1] (block2) {$2$}
	node [block, right = 0.0005 cm of block0] (block3) {$0$}

	node [block, right = 1 cm of block3] (block0){$0$}
    node [block, below = 0.0005 cm of block0] (block1) {$1$}
    node [block, below = 0.0005 cm of block1] (block2) {$2$}
	node [block, right = 0.0005 cm of block0] (block3) {$1$}
	
	node [block, right = 1 cm of block3] (block0){$0$}
    node [block, below = 0.0005 cm of block0] (block1) {$1$}
    node [block, below = 0.0005 cm of block1] (block2) {$2$}
	node [block, right = 0.0005 cm of block0] (block3) {$2$}
    ;
\end{tikzpicture}
}
\caption{\small Young tableaux of shape $\lambda = (2, 1, 1)$ leads to three different distribution for $T = \left\lbrace (2, 1, 1), (1, 2, 1), (1, 1, 2)\right\rbrace$ which helps to obtain $\mathcal{S}_{\lambda} \left( D, D^2, D^4 \right)$}
\label{schur}
\end{figure}
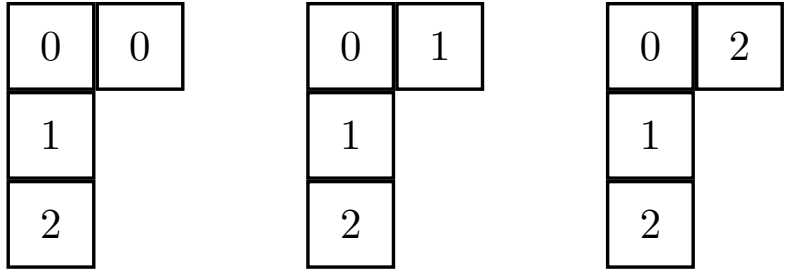 

%% file: supp_num_exp.tex
\subsection{Search Time for Random Convolutional Coding}
\label{sec:more_num_exp}

We run an experiment to tabulate the time needed to find a good random matrix $\bfR$. We run $50$ trials to find the best $\bfR$ for $n = 13, 14, 15$ with $s = 2, 3, 4$. It should be noted that the choice of $\bfR$ depends on all $\binom{n}{s}$ choices of stragglers. Fig. \ref{strtime234} shows the corresponding time for different pairs of $n$ and $s$. From the figure, it can be seen that our system (a processor with CPU speed $3.5 GHz$ and $16 GB$ RAM) needs only around $8$ minutes to find a good choice of $\bfR$ for even $n = 15$ and $s = 4$. In other cases, the required amount of time is even lesser. This indicates that for a reasonable system size, we do not need to wait too long to obtain a good choice of $\bfR$ that ensures that the worst case condition number is bounded. And it should be noted that this is a one-time cost for designing the coding scheme for a system with $n$ worker nodes which is resilient to $s = n - k$ stragglers.

\begin{figure}[t]
\centering
\captionsetup{justification=centering}
\resizebox{0.65\linewidth}{!}{
\begin{tikzpicture}
\begin{axis}[
width=5in,
height=3.203in,
at={(2.6in,0.852in)},
major x tick style = transparent,
ybar=2*\pgflinewidth,
bar width=20pt,
ymajorgrids,
xmajorgrids,
xlabel style={font=\color{white!15!black}, font = \Large},
ylabel style={font=\color{white!15!black}, font = \Large},
ylabel={Required time (in seconds)},
symbolic x coords={{\LARGE s = 2},{\LARGE s = 3},{\LARGE s = 4}},
xtick = data,
scaled y ticks = false,
enlarge x limits= 0.3,
ymin=0,
ymax=550,
legend cell align=left,
legend style={at={(0.02,0.63)}, nodes={scale=1.6}, anchor=south west, legend cell align=left, align=left, draw=white!15!black}
    ]
    \addplot[style={fill=my1color,mark=none}]
            coordinates {({\LARGE s = 2}, 29.96) ({\LARGE s = 3},86.98) ({\LARGE s = 4}, 195.04)};
\addlegendentry{$n = 13$}
   \addplot[style={fill=mycolor1,mark=none}]
            coordinates {({\LARGE s = 2},48.98) ({\LARGE s = 3},135.49) ({\LARGE s = 4},307.30)};
\addlegendentry{$n = 14$}
    \addplot[style={fill=mycolor2,mark=none}]
             coordinates {({\LARGE s = 2},73.30) ({\LARGE s = 3},194.02) ({\LARGE s = 4},485.24)};
             \addlegendentry{$n = 15$}
    \end{axis}

\end{tikzpicture}%
}
\caption{Comparison of required time to find a good choice of $\bfR$ for different $n$ and $s$.}
\label{strtime234}
\end{figure}
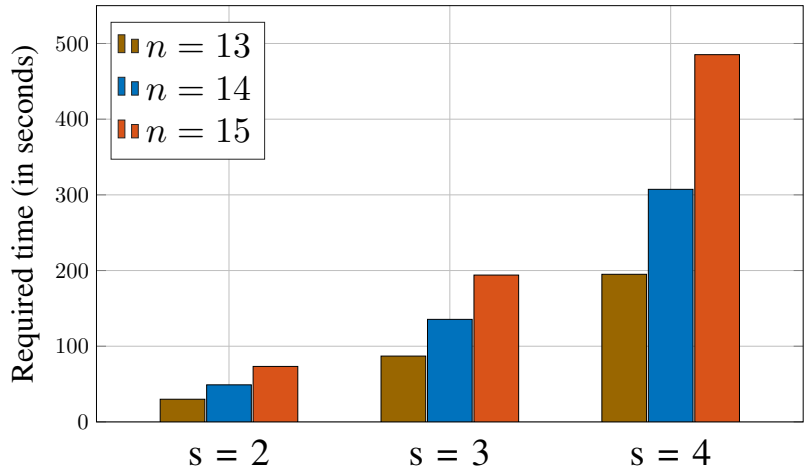